\documentclass[a4paper,aps,amsmath,amssymb,onecolumn,nofootinbib,pra,superscriptaddress,accepted=2024-05-13]
{quantumarticle}
\pdfoutput=1
\usepackage[utf8]{inputenc}
\usepackage[english]{babel}
\usepackage[T1]{fontenc}
\usepackage{amsmath}
\usepackage{hyperref}

\usepackage{tikz}
\usepackage{lipsum}

\usepackage{graphicx,tikz}
\usepackage{dcolumn}
\usepackage{bm}
\usepackage{epsfig,hyperref,amsthm}
\usepackage{mathtools}
\usepackage{color,fancybox}
\usepackage{colordvi}
\usepackage{relsize}
\usepackage{accents}
\usepackage{wasysym} 
\usepackage{xypic}
\usepackage{enumitem}
\usepackage{mathtools,slashed}

\hypersetup{
	colorlinks=true,
	linkcolor=CitingColor,
	citecolor=CitingColor,
	urlcolor=CitingColor
}

\DeclareMathAlphabet\mathbfcal{OMS}{cmsy}{b}{n}

\newcommand{\rms}{{\rm S}}
\newcommand{\rmt}{{\rm T}}
\newcommand{\rmx}{{\rm X}}
\newcommand{\rhos}{\rho_\rms}

\newcommand{\rhox}{\rho_{\rm X}}

\newcommand{\Scr}{{\rm Score}}

\newcommand{\bml}{{\bm{\Lambda}}}

\newcommand{\CRTL}{\mathfrak{C}_{R|\rmt,\bml}}
\newcommand{\CL}{\mathfrak{C}_{\bml}}
\newcommand{\RT}{{R|\rmt}}
\newcommand{\ket}[1]{\ensuremath{|#1\rangle}}
\newcommand{\bra}[1]{\ensuremath{\langle #1|}}

\newcommand{\proj}[1]{\ket{#1}\!\bra{#1}}
\newcommand{\be}{\begin{equation}}
\newcommand{\ee}{\end{equation}}
\newcommand{\ba}{\begin{eqnarray}}
\newcommand{\ea}{\end{eqnarray}}

\newcommand{\oge}{\succeq}
\newcommand{\ole}{\preceq}

\newcommand{\tr}{{\rm tr}}
\newcommand{\ttr}{{\rm Tr}}

\newcommand{\norm}[1]{\left\|#1\right\|}

\newcommand{\id}{\mathbb{I}}

\newcommand{\bmsl}{\bm{\sigma}_{\bml}}

\newcommand{\Free}{\mathcal{F}_{R}} 
\newcommand{\Fr}[1]{\mathcal{F}_{R|#1}} 
\newcommand{\FRS}{\Fr{\rm S}} 
\newcommand{\FRT}{\Fr{\rm T}} 

\newcommand{\SRT}{\mathcal{S}_{\RT}} 
\newcommand{\etaT}{\eta_{\rm T}}
\newcommand{\RRT}{\mathcal{R}_{R|\rmt}}

\newcommand{\calB}{\mathcal{B}}
\newcommand{\mE}{\mathcal{E}}

\newcommand{\E}{\mathcal{E}}
\newcommand{\C}{\mathcal{C}}
\renewcommand{\H}{\mathcal{H}}

\newcommand{\calL}{\mathcal{L}}

\newcommand{\U}{\mathcal{U}}

\newcommand{\rT}{{\rm T}}
\newcommand{\rX}{{\rm X}}

\newcommand{\EIL}{\mathbfcal{E}_{\bm{I},\bm{\Lambda}}}

\newtheorem{theorem}{Theorem}
\newtheorem{corollary}[theorem]{Corollary}

\newtheorem{proposition}[theorem]{Proposition}
\newtheorem{alemma}{Lemma}[section]

\newtheorem{atheorem}[alemma]{Theorem}

\newtheorem{afact}[alemma]{Fact}

\newtheorem{question}{Question}

\newtheorem{assumption}{Assumptions}

\definecolor{nred}{rgb}{0.9,0.1,0.1}
\definecolor{nblack}{rgb}{0,0,0}
\definecolor{nblue}{rgb}{0.2,0.2,0.8}
\definecolor{ngreen}{rgb}{0.3,0.7,0.3}
\definecolor{ublue}{rgb}{0,0,0.5}

\definecolor{pur}{rgb}{0.75,0,0.75}
\definecolor{nngrn}{rgb}{0,0.5,0.5}
\definecolor{CitingColor}{rgb}{0,0.3,1}

\newcommand{\blu}{\color{nblue}}

\usepackage[capitalize]{cleveref}

\begin{document}
\title{Resource Marginal Problems}

\author{Chung-Yun Hsieh}
\email{chung-yun.hsieh@bristol.ac.uk}
\affiliation{H. H. Wills Physics Laboratory, University of Bristol, Tyndall Avenue, Bristol, BS8 1TL, UK}
\affiliation{ICFO - Institut de Ci\`encies Fot\`oniques, The Barcelona Institute of Science and Technology, 08860 Castelldefels, Spain}

\author{Gelo Noel M. Tabia}
\affiliation{Foxconn Quantum Computing Research Center, Taipei 114, Taiwan}
\affiliation{Department of Physics and Center for Quantum Frontiers of Research \& Technology (QFort), National Cheng Kung University, Tainan 701, Taiwan}
\affiliation{Physics Division, National Center for Theoretical Sciences, Taipei 106319, Taiwan}
\affiliation{Center for Quantum Technology, National Tsing Hua University, Hsinchu 300, Taiwan}

\author{Yu-Chun Yin}
\affiliation{Institute of Communications Engineering, National Yang Ming Chiao Tung University, Hsinchu 300093, Taiwan}
\affiliation{Department of Physics and Center for Quantum Frontiers of Research \& Technology (QFort), National Cheng Kung University, Tainan 701, Taiwan}

\author{Yeong-Cherng Liang}
\email{ycliang@mail.ncku.edu.tw}
\affiliation{Department of Physics and Center for Quantum Frontiers of Research \& Technology (QFort), National Cheng Kung University, Tainan 701, Taiwan}
\affiliation{Physics Division, National Center for Theoretical Sciences, Taipei 106319, Taiwan}


\begin{abstract}
We introduce the {\em resource marginal problems}, which concern the possibility of having a resource-free target subsystem compatible with a {\em given} collection of marginal density matrices. 
By identifying an appropriate choice of resource $R$ and target subsystem $\rmt$, our problems reduce, respectively, to the well-known {\em marginal problems} for quantum states and the problem of determining if a given quantum system is a resource.
More generally, we say that a set of marginal states is {\em resource-free incompatible} with a target subsystem $\rmt$ if all global states compatible with this set must result in a resourceful state in $\rmt$ of type $R$. 
We show that this incompatibility {\em induces} a resource theory that can be quantified by a monotone 
and obtain necessary and sufficient conditions for this monotone to be computable as a conic program with finite optimum.
We further show, via the corresponding witnesses, that (1) resource-free incompatibility is equivalent to an operational advantage in some channel-discrimination tasks, and (2) some specific cases of such tasks fully characterize the convertibility between marginal density matrices exhibiting resource-free incompatibility.
Through our framework, one sees a clear connection between any marginal problem---which implicitly involves some notion of incompatibility---for quantum states and a resource theory for quantum states. 
We also establish a close connection between the physical relevance of resource marginal problems and the ground state properties of certain many-body Hamiltonians.
In terms of application, the universality of our framework leads, for example, to a further quantitative understanding of the incompatibility associated with the recently-proposed entanglement marginal problems and entanglement transitivity problems. 
\end{abstract}

\maketitle

\section{Introduction}\label{Sec:Intro}
In quantum information theory, a question that has drawn wide interest is determining if a given set of density matrices is compatible, i.e., whether some global state gives this collection of density matrices as its marginals.
Such problems are known collectively as  {\em marginal problems}~\cite{Higuchi2003,Bravyi2004,Klyachko-Review,Schilling-Thesis} for quantum states. 
Originally, they were motivated by the computation of the ground states of 2-body, usually local, Hamiltonians. 
This gives the well-known $2$-body $N$-representability problem, which asks whether the given $2$-body states can be the marginals of a single $N$-body state (see, e.g., Refs.~\cite{Klyachko-Review,bookchapter}).
A more refined version of such problems, which further demands that the global state is entangled, was comprehensively discussed in the recent work of Ref.~\cite{Navascues2021} (see also Refs.~\cite{Toth2007,Navascues2017}).

Indeed, quantum entanglement~\cite{Ent-RMP} has long been recognized as a resource under the paradigm of local operations assisted by classical communications. 
Over the years, this resource-theoretic viewpoint has further sparked the development of other {\em resource theories}~\cite{RT-RMP,Coherence-RMP,Brandao2013,Gour2008,deVicente2014,Gallego2015,Lostaglio2018,Wolfe2019}, aiming, e.g., for the quantification ~\cite{Vedral1997,Skrzypczyk2014,Baumgratz2014,Chruscinski2014,CGL15,Marvian2016} of resources and their inter-convertibility~\cite{Marvian2013,Brandao2015}.
For quantum states (and correlations), examples of such resources include, but not limited to, entanglement~\cite{Ent-RMP,Vedral1997}, coherence~\cite{Coherence-RMP,Baumgratz2014}, athermality~\cite{Horodecki2013,Serafini2019,Narasimhachar2019}, asymmetry~\cite{QRF-RMP,MarvianThesis}, nonlocality~\cite{Bell-RMP}, and steering~\cite{Steering-RMP,steering-review}. 
Thanks to the generality of the resource-theoretic framework, several structural features shared by many resources of states~\cite{Regula2018,Liu2017,Anshu2018,Liu2019,Fang2019,Takagi2019,Takagi2019-2,Regula2020,Sparaciari2020,Uola2019} have been made evident.

Although marginal problems and resource theories are seemingly unrelated, several recent developments have made clear that it is fruitful to consider them {\em simultaneously}.
For instance, the entanglement properties of a global pure state can be deduced from the spectrum of its single-party density matrices~\cite{Walter2013}. 
Even without the pure-state assumption, the nonlocality (and hence entanglement) of certain $N$-body systems can be certified using {\em only} its two-body marginals~\cite{Tura2014}. 
In fact, certain marginal information may already be sufficient to certify the entanglement~\cite{Tabia} and nonlocality of some {\em other} subsystems~\cite{Bancal2012,Barnea2013} (see also Ref.~\cite{Coretti2011}). 

As an example, consider the recently proposed {\em entanglement transitivity problems}~\cite{Tabia}.
In a tripartite setting ABC, an entanglement transitivity problem asks whether there exist entangled bipartite states $\rho_{\rm AB}$ in AB and $\sigma_{\rm BC}$ in BC such that 
\begin{itemize}
\item they are compatible; that is, they are the reduced states of some tripartite global states;
\item {\em all} these tripartite states {\em must} have their AC marginal as some {\em entangled} states.
\end{itemize}
Whenever the answer to this problem is positive, the marginal information in AB and BC {\em necessarily implies} the entanglement in AC. An explicit example~\cite{Tabia} of such entanglement-implying marginals are the bipartite marginals of the $W$-state~\cite{Dur2000}.
Surprisingly, there also exist solutions in which both $\rho_{\rm AB},\sigma_{\rm BC}$ are {\em separable} while AC is again forced to be entangled
(dubbed {\em meta-transitivity} in Ref.~\cite{Tabia}).

All these recent advances suggest the importance and need to explore various variants of marginal problems, say, involving different quantum state resources, within a single theoretical framework.
In particular, could one arrive at some general conclusions without first specifying the resource of interest? 
Here, we answer these questions in the positive by providing the first unified framework that naturally incorporates state resource theories into the state marginal problems.

\section{Preliminary Notions}\label{Sec:Preliminary}

\subsection{Quantum Resource Theories: a Brief Introduction}\label{Sec:QRT}
In the study of quantum information science, once a new resource has been discovered, such as entanglement, coherence, etc., it is important to ask the following three questions:
\begin{enumerate}[label = (Q\arabic*)]
	\item\label{Q:Characterisation} {\em How to characterize it?} That is, can we find necessary and sufficient conditions for the existence of this resource so that we can somehow detect it? 
	\item\label{Q:Comparison} {\em How to compare its contents?} Sometimes, merely knowing the existence of the resource is not enough. 
 For instance, a perfect performance in teleportation~\cite{BBC+93} not only needs entanglement, but also a maximally entangled one.
Hence, when we have two states both equipped with the resource, it is important to know how to compare them.
	\item\label{Q:Quantification} {\em How to quantify it?} Also, it is useful to know whether the strength of this resource can be evaluated in a quantitative way.
\end{enumerate}
It turns out that {\em quantum resource theories} provide a systematic, general way to do these.
We now recall the main ingredients of a quantum resource theory, or simply a resource theory. 
For further details, see, e.g., Ref.~\cite{RT-RMP}.
Formally, a resource theory for quantum states is specified by a triplet $(R,\mathcal{F}_R,\mathcal{O}_R)$, where $R$ represents the given resource. 
$\mathcal{F}_R$ is the set of {\em (resource-) free states}, which contains all states that {\em do not} possess the given resource $R$.
Finally, $\mathcal{O}_R$ is the set of {\em free operations}, which are physical transformations of states that {\em cannot} generate the given resource $R$.
For instance, when we aim to study entanglement in a given bipartition (e.g., specified by two agents A and B), we can consider
\begin{align}
	R& = \text{entanglement}\nonumber\\
	\mathcal{F}_R& = \text{separable states}\nonumber\\
	\mathcal{O}_R& = \text{local operations and classical communications (LOCC)}\nonumber
\end{align} 
Note that, throughout, we allow $\mathcal{F}_R$ and $\mathcal{O}_R$ to be of {\em any finite} dimension.
For entanglement, $\mathcal{F}_R$ contains separable states with {\em all} possible system sizes in the given bipartition AB.
They can be two-qubit separable states shared by A and B, an $(N+1)$-qubit separable states where A holds $N$ qubits, B holds one, etc.
Similarly, $\mathcal{O}_R$ contains LOCC acting separately on all the subsystems held by each agent.

Of course, resource theories are not only for entanglement.
Let us consider another example, which is the resource theory of non-equilibrium thermodynamics:
\begin{align}
	R& = \text{athermality}\nonumber\\
	\mathcal{F}_R& = \text{thermal states}\nonumber\\
	\mathcal{O}_R& = \text{thermal operations}\nonumber
\end{align} 
In this case, the resource is the status of being out-of-equilibrium (i.e., {\em athermality}), and the free states are thermal equilibrium states associated with the given system Hamiltonian and background temperature.
Finally, the allowed dynamics are thermal operations that cannot drive such states out of thermal equilibrium (see, e.g., Refs.~\cite{Horodecki2013,Serafini2019,Narasimhachar2019} for details).

In general, once we are able to define $\mathcal{F}_R$ and $\mathcal{O}_R$, a resource theory can be formulated to study the corresponding resource $R$.
If we start by setting $\mathcal{F}_R$, then we identify all states in $\mathcal{F}_R$ as {\em free}.
Hence, a suitable characterization of this mathematical set can allow us to characterize states {\em outside} this set, which can further answer question~\ref{Q:Characterisation}.

Note that once we choose $\mathcal{F}_R$, the structure of $\mathcal{O}_R$ cannot be arbitrary.
Internal consistency of a resource theory demands that any free operation $\mE$ acting on a free state $\eta$ {\em cannot} generate a resource state~\cite{RT-RMP}, hence 
\begin{align}
\mE(\eta)\in\mathcal{F}_R\quad\;\forall\; \eta\in\mathcal{F}_R\;\&\;\forall\;\mE\in\mathcal{O}_R.
\end{align}
This is sometimes called the {\em golden rule}~\cite{RT-RMP}, which is a central necessary condition in every resource theory.
With this golden rule, channels in $\mathcal{O}_R$ are understood to have {\em no} ability to generate $R$.
Now, if we can use a free operation $\mathcal{E}\in\mathcal{O}_R$ to map a state $\rho$ to another $\sigma$; i.e., $\mathcal{E}(\rho) = \sigma$, then {\em by definition} the output {\em cannot} be more resourceful than the input, since free operation cannot generate $R$.
Hence, a suitable characterization of convertibility under free operations can help us to {\em compare} the resource content of two different states, thereby answering question~\ref{Q:Comparison}.

Finally, for the quantification of $R$ [i.e., question~\ref{Q:Quantification}], one makes use of  a {\em resource monotone} $Q_R$, which satisfies 
\begin{itemize}
\item$Q_R(\rho)\ge0$ where equality holds if $\rho\in\mathcal{F}_R$. 
\item$Q_R[\mE(\rho)]\le Q_R(\rho)\;\forall\;\rho$ and $\forall\;\mE\in\mathcal{O}_R$.
\end{itemize}
Every $Q_R$ satisfying these two conditions can be understood as an appropriate quantifier of the resource $R$, thereby answering question~\ref{Q:Quantification}.
While one can impose more axioms on the definition of $Q_R$, we shall keep only the minimal requirements in this work.

Finally, note that although we only mention resource theories of quantum states, this general approach is not restricted to states.
One may also consider resource theories for measurements~\cite{Paul,Paul2022}, channels~\cite{Rosset2018,Bauml2019,LiuWinter2019,LiuYuan2019,Bu2020,Ku2022PRXQ,Vieira2024,Stratton2023}, state assemblages~\cite{Gallego2015,Skrzypczyk2014,Hsieh2016,Quintino2016,Hsieh2023,Ku2022,Ku2023,Hsieh2024,Hsieh2024-2}, measurement assemblages~\cite{Ku2022,Hsieh2023,Ku2023,Buscemi2020,Buscemi2023,Skrzypczyk2019}, etc.---as long as the mathematical properties can be well-defined and analyzed.

\subsection{Marginal Problems}
Next, let us recall the marginal problem for quantum states. 
Consider a finite-dimensional $n$-partite global system $\rms$ and let ${\bf S}$ be the collection of all $2^n-1$ nontrivial combinations of the subsystems of $\rms$. 
Moreover, let $\bm{\Lambda}$ be a collection of such subsystems, i.e., $\bm{\Lambda}\subseteq {\bf S}$. Obviously, any combination ${\rm X\in}\bm{\Lambda}$ of subsystems satisfies ${\rm X\in {\bf S}}$.
We say that a set of marginal states $\bmsl\coloneqq\{\sigma_{\rm X}\}_{{\rm X\in}\bm{\Lambda}}$  indexed by $\bm{\Lambda}$ is {\em compatible} if there exists a global state $\rhos$ such that each $\sigma_{\rm X}$ is recovered by performing the respective partial trace on $\rhos$:
\begin{align}
\tr_{\rm S\setminus X}(\rhos) = \sigma_{\rm X}\quad\;\forall\;{\rm X\in}\bm{\Lambda};
\end{align}
the corresponding $\rhos$ is then said to be compatible with $\bmsl$. When there is no such $\rhos$, $\bmsl$ is said to be {\em incompatible}.
For the above setting, the corresponding {\em marginal problem} consists in answering the following question:
\begin{center}
{\em
Is $\bmsl$ compatible?
}
\end{center}
As mentioned in \cref{Sec:Intro}, this type of problems are relevant to different aspects of quantum information science, many-body physics, and quantum chemistry.
An important motivation for considering such problems lies on understanding when the parts indeed determine the whole, see, e.g.,~\cite{LPW02,LW02,Yu2021}.

One can also consider marginal problems for probability distributions (which can be viewed as classical marginal problems). Importantly, there are fundamental differences between classical marginal problems and the marginal problems for quantum states. 
Consider, for example, the tripartite setting ABC. 
In this setting, if two bipartite probability distributions $\{P_{\rm AB}(ab)\},\{P_{\rm BC}(bc)\}$ are {\em locally compatible}, in the sense that $\sum_aP_{\rm AB}(ab) = P_{\rm B}(b) = \sum_cP_{\rm BC}(bc)$, then they {\em must} be compatible, since (see, e.g., Ref.~\cite{Hsieh2021})
\begin{align}
	P_{\rm ABC}(abc)\coloneqq\frac{P_{\rm AB}(ab)P_{\rm BC}(bc)}{P_{\rm B}(b)},
\end{align}
is a valid probability distribution satisfying
$\sum_aP_{\rm ABC}(abc) = P_{\rm BC}(bc)$
and 
$\sum_cP_{\rm ABC}(abc) = P_{\rm AB}(ab)$.
Hence, in this setting, a classical marginal problem is always trivial since local compatibility implies global compatibility.
Nevertheless, in the quantum case, this simplest setting is already {\em non-trivial}---one can consider maximally entangled states in AB and BC, which obviously have the same marginal state in B.
However, since entanglement is monogamous, it is impossible to find a tripartite state to realize these reduced states at the same time.
This demonstrates the gap between classical and quantum marginal problems, which further shows that the involvement of quantum resources can considerably change the structure of marginal problems.

\subsection{Conic Programming} 
We now briefly review {\em conic programming}.
Rather than being mathematically general, we focus on a specific form relevant to this work's physical setting (similar to the one used in, e.g., Ref.~\cite{Uola2019}).

To start with, consider a vector space $\mathcal{H}$ describing matrices with a given finite dimension.
For instance, it can be the space of all $2\times 2$ matrices, which includes all qubit density matrices.
In this vector space, a {\em proper cone} denoted by $\C\subseteq\mathcal{H}$ is a nonempty, convex and closed subset of $\H$ such that 
\begin{itemize}
\item $\alpha x\in\C$ for every $\alpha\ge0$ if $x\in\C$.
\item $x\in\C$ and $-x\in\C$ imply that $x$ is the null vector in $\H$.
\end{itemize}
Roughly speaking, the first condition means that $\C$ is a collection of ``rays'' emitted from the origin, which can be visualized as a cone.
On the other hand, the second condition means that this cone has a preferred direction, since it cannot contain a line passing through the origin.
With a given proper cone $\C$, we can define a partial order---we can {\em compare} $x,y$ via $\C$ by defining $x\preceq_{\C} y$ whenever $y-x\in\C$.
In this work, $x\preceq y$ means that $y-x$ is a positive semi-definite matrix, which can be viewed as the partial order defined by the positive semi-definite cone.

Now, consider two vector spaces $\H,\H'$ of matrices with possibly different finite dimensions.
Suppose that they are equipped with some inner products, which we write as $\langle \cdot,\cdot\rangle$ for simplicity.
In this work, we consider the Hilbert-Schmidt inner product or some other forms induced by it.
Then, following Refs.~\cite{Uola2019,Uola2020,Takagi2019,Gartner2012,Boyd:Book}, the {\em primal problem} of a conic program can be written as~\cite{Gartner2012}
\begin{equation}\label{Eq:CPPrimal}
\begin{split}
\max_{x}\quad&\langle A,x\rangle\\
{\rm s.t.}\quad& x\in\C;\;\mathfrak{L}(x)\preceq B,
\end{split}
\end{equation}
where $\mathfrak{L}:\mathcal{H}\to\mathcal{H}'$ is a linear map, and $A\in\H,B\in\H'$ are some constants.
The (Lagrange) {\em dual problem} of Eq.~\eqref{Eq:CPPrimal} may then be written as~\cite{Gartner2012}:
\begin{equation}\label{Eq:CPDual}
\begin{split}
\min_{z}\quad&\langle B,z\rangle\\
{\rm s.t.}\quad&z\succeq 0;\;\left\langle \mathfrak{L}^\dagger(z)-A,x\right\rangle\ge0\;\forall\;x\in\C,
\end{split}
\end{equation}
where $\mathfrak{L}^\dagger$ is the map dual to $\mathfrak{L}$. 
By construction, the optimum of \cref{Eq:CPDual} always upper bounds that of \cref{Eq:CPPrimal}.
When these optimum values coincide, one says that {\em strong duality} holds~\cite{Boyd:Book}.
This happens, for example, when the primal problem is convex and Slater's conditions hold. 
In our setting, it is equivalent to checking whether there exists a point $x$ in the relative interior of $\C$ such that $B-\mathfrak{L}(x)$ is strictly positive (see Refs.~\cite{Uola2019,Gartner2012}, Appendices~\ref{App:MathPreliminary}, and~\ref{App:CPIntro} for further details).

\section{Results}   
\subsection{Framework: Resource-Free Compatibility} 

Importantly, as will become evident from our framework, the problem of whether a given physical system $\rmt$ is a resource of type $R$ can also be phrased as a compatibility problem. 
Our interest is to provide a unified framework for addressing both types of compatibility problems at the same time. 
Henceforth, we adopt the shorthand $\rho_{\rm X}\coloneqq\tr_{\rm S\setminus X}(\rhos)$ and denote by $\FRT$ the set of all free states with respect to the resource $R$ in the subsystem $\rmt$, i.e., $\FRT = \mathcal{F}_R\cap\mathcal{S}_{\rm T}$, where $\mathcal{S}_{\rm T}$ is the set of all states on $\rmt$. 
A central notion capturing both compatibilities is then given as follows.
For any given subsystem $\rmt$ of $\rms$ and resource of type $R$, the collection of density matrices $\bmsl$ is said to be {\em $R$-free compatible in $\rmt$} (or simply {\em $R$-free compatible} when there is no risk of confusion) if 
\begin{align}
\exists\;\rhos\;\;{\rm compatible\;with}\;\;\bmsl\;{\rm s.t.}\;\rho_\rmt\in\FRT.
\end{align}
Namely, if $\bmsl$ is $R$-free compatible, one can find a global state $\rhos$ such that all density matrices in $\bmsl$ are its marginals and, in particular, its marginal in $\rmt$ is {\em free}.
In other words, density matrices in $\bmsl$ are not only compatible among themselves, but also with a target system T that is {\em not} a quantum resource of type $R$.

Conversely, $\bmsl$ is called {\em $R$-free incompatible in $\rmt$} (or simply {\em $R$-free incompatible}) if it does not satisfy the above condition, namely, either $\bmsl$ is incompatible or
\begin{align}
\forall\;\rhos\;\;{\rm compatible\;with}\;\;\bmsl\;\Rightarrow\;\rho_\rmt\notin\FRT.
\end{align}
When $\bmsl$ is compatible {\em and} $R$-free incompatible at the same time, one can use these density matrices to {\em certify} the resourceful nature of the target system---since {\em every} global state having $\bmsl$ as its marginals {\em must} possess the quantum resource of type $R$ in T.
This is the notion of resource that we aim to study in this work.

Our central question is then defined as follows:
For any given triplet ($R$, $\bmsl$, $\rmt$), the corresponding {\em resource marginal problem} consists in answering the following question:
\begin{center}
{\em
Is $\bmsl$ $R$-free compatible in $\rmt$?
}
\end{center}
To answer this question, consider now, for the given triplet ($R$, $\bm{\Lambda}$, $\rmt$), the set of density matrices indexed by $\bm{\Lambda}$  that are $R$-free compatible in $\rmt$:
\begin{align}\label{Eq:CTR}
	\CRTL\coloneqq\{\bm{\tau}_{\bm{\Lambda}}\,|\,\exists\,\rhos\;{\rm compatible\;with}\;\bm{\tau}_{\bm{\Lambda}}\;{\rm s.t.}\;\rho_\rmt\in\FRT\}.
\end{align}
Clearly, the set $\CL$ of {\em all} compatible marginal density matrices associated with subsystems specified by $\bm{\Lambda}$ is a superset of $\CRTL$.
Then, the complement of $\CRTL$ in $\CL$, i.e., $\CL\setminus\CRTL$, is simply the set of compatible $\bmsl$ that must necessarily result in a resourceful marginal in $\rmt$.

Among others,  resource marginal problems contain as special cases the usual marginal problem for quantum states and the problem of deciding whether a given system is a resourceful state. 
Of course, our unified framework allows us to obtain a general, quantitative analysis beyond these special cases, see~\cref{Fig:Summary}.
Note further that assuming $\bmsl\in\CL$ may be very natural, e.g., in certifying the resource nature of $\rmt$. 
However, this assumption does not have to be imposed {\em a priori} since the characterization of $\CL$ represents a specific instance of our general problem. Indeed, checking the membership of $\CL\setminus\CRTL$ is equivalent to determining the membership of both $\CL$ and $\CRTL$. Hence, we shall stick to the most general setting given by the definition of $R$-free compatibility in subsequent discussions.

\begin{center}
\begin{figure*}
\scalebox{0.85}{\includegraphics{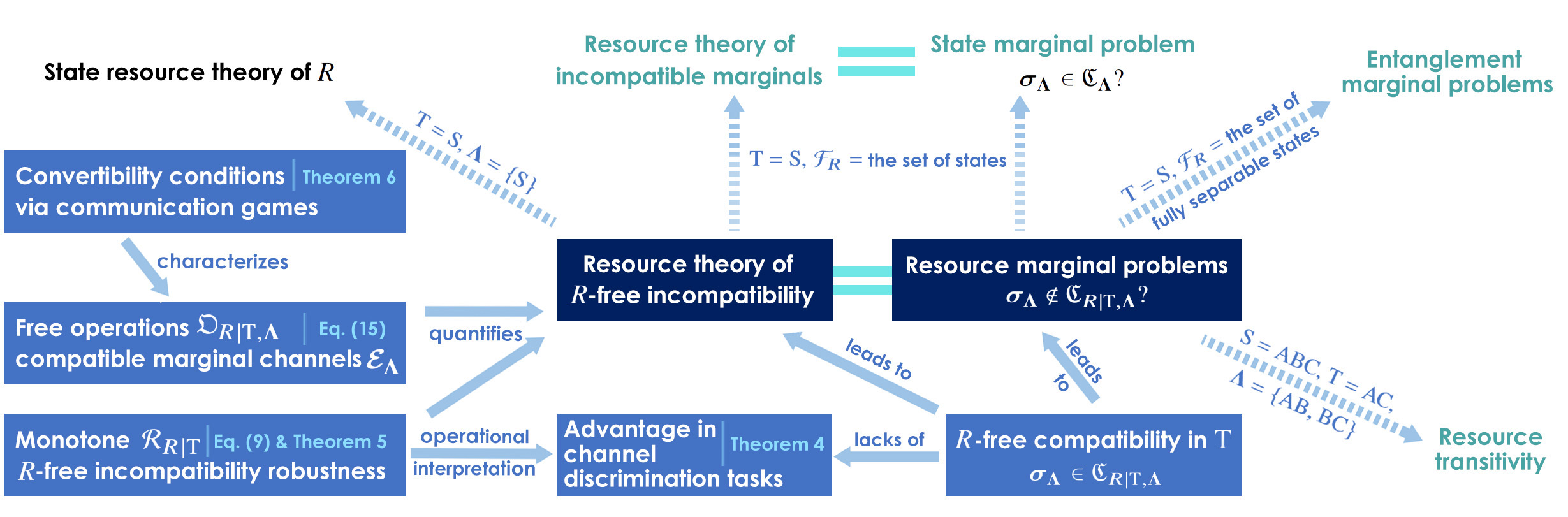}}
\caption{
{\bf Summary of the key notions introduced (boxed) in this work and how they are related to each other.}
For a given collection of density matrices $\bmsl$ for subsystems $\bm{\Lambda}$, the set of compatible $\bmsl$ is denoted by $\CL$. 
When a resource $R$ and a target subsystem $\rmt$ are also specified, one can further define $\CRTL\subseteq\CL$ such that all members of $\CRTL$ are $R$-free compatible in $\rmt$. 
This primitive notion immediately leads to the resource marginal problem corresponding to a given $R$, $\rmt$, and $\bmsl$. 
When $\bmsl\in\CRTL$ is taken as a free resource and the set $\mathfrak{O}_{R|\rmt,\bm{\Lambda}}$ of compatible marginal channels $\mathbfcal{E}_{\bm{\Lambda}}$, \cref{Eq:FreeOperation}, is taken as the set of free operations, then we also obtain the resource theory of $R$-free incompatibility for the collection of density matrices $\bmsl$. 
A dashed line originating from a node leads to a special case of this node. 
Double solid lines without any arrow connecting two nodes means that a one-to-one correspondence can be established between the two.
}
\label{Fig:Summary} 
\end{figure*}
\end{center}

\subsection{$R$-Free Incompatibility as a Resource}
Inspired by the approach of Refs.~\cite{Uola2019,Uola2020}, we now introduce the {\em $R$-free incompatibility robustness} $\RRT(\bmsl)$
to quantify resource-free incompatibility. Formally, for every set of states $\bmsl$, define
\begin{align}\label{Eq:Robustness}
	\RRT(\bmsl)\coloneqq\log_2\inf\left\{\lambda\ge0\,|\,\bmsl\preceq\lambda\bm{\kappa}_{\bm{\Lambda}}, \bm{\kappa}_{\bm{\Lambda}}\in\CRTL\right\},
\end{align}
where the inequality $\bmsl\preceq\lambda\bm{\kappa}_{\bm{\Lambda}}$ means $\sigma_{\rm X}\preceq\lambda\kappa_{\rm X}$ for every ${\rm X}\in{\bm\Lambda}$.
To understand the meaning of this figure-of-merit, note that it can be {\em equivalently} defined as the minimization over a real number $\mu\ge0$ such that
\begin{align}\label{Eq:Robustness-v2}
	2^{-\mu}\bmsl + \left(1-2^{-\mu}\right)\bm{\tau}_{\bm{\Lambda}}\in\CRTL,
\end{align}
where $\bm{\tau}_{\bm{\Lambda}}$ is any collection of density matrices, playing the role of {\em noise} in the above convex mixture.
In this form, one can see that $\RRT$ measures ``how robust'' the $R$-free incompatibility of the given collection $\bmsl$ is against noise described by the term $\bm{\tau}_{\bm{\Lambda}}$.
Importantly, there is no restriction on the noise $\bm{\tau}_{\bm{\Lambda}}$ here, and for this reason this type of robustness is sometimes called {\em generalized} robustness.

Evidently, $\RRT(\bmsl)\ge0$, and equality holds if $\bmsl~\in~\CRTL$.
This means that $\RRT>0$ certifies $R$-free incompatibility.
When the set $\CRTL$ is closed (in a suitably chosen topology), $\RRT(\bmsl)=0$ {\em only if} $\bmsl\in\CRTL$; in this case, the robustness measure becomes {\em faithful}.
Now, for a given set of operators ${\bf O}_{\bm{\Lambda}}\coloneqq\{O_{\rm X}\}_{\rm X\in{\bm\Lambda}}$, we define $\langle\bmsl,{\bf O}_{\bm{\Lambda}}\rangle\coloneqq\sum_{{\rm X\in}\bm{\Lambda}}\tr(\sigma_{\rm X}O_{\rm X})$.
Following the methodologies of Refs.~\cite{Uola2019,Uola2020,Takagi2019}, some general results can be obtained by imposing the following assumptions on $R$, $\rmt$, and $\bm{\Lambda}$:

\newpage

\renewcommand*{\theassumption}{\Alph{assumption}}

\begin{assumption}\label{Assumptions}
{\em\bf (Working Hypothesis)}
\begin{enumerate}[label = (A\arabic*)]
\item\label{Assumption:Convex} $\FRT$ is convex and compact.
\item\label{Assumption:Static} There exists $\rhos$ such that $\rho_\rmt\in\FRT$ and $\rho_{\rm X}$ is full-rank for every ${\rm X\in}\bm{\Lambda}$.
\end{enumerate}
\end{assumption}
Specifically, we prove the following in~\cref{App:Proof-Result:RTWitness}.

\begin{theorem}\label{Result:RTWitness}
{\em\bf ($R$-Free Incompatibility Witness)}
Let $R$, {\rm T}, and $\bm{\Lambda}$ satisfy Assumptions~\ref{Assumptions}, then $\bmsl\notin\CRTL$ if and only if there exists ${\bf W}_{\bm{\Lambda}}\coloneqq\{W_{\rm X}\succeq0\}_{{\rm X\in}\bm{\Lambda}}$ such that
\begin{align}
\sup_{\bm{\tau}_{\bm{\Lambda}}\in\CRTL}\langle\bm{\tau}_{\bm{\Lambda}},{\bf W}_{\bm{\Lambda}}\rangle<\langle\bmsl,{\bf W}_{\bm{\Lambda}}\rangle.
\end{align}
\end{theorem}
Theorem~\ref{Result:RTWitness} can be understood as the witness of $R$-free incompatibility, i.e., the $R$-free incompatibility of any given $\bmsl$ can always be detected by a separating hyperplane formed by finitely many positive semi-definite $W_{\rm X}$.
It is also worth mentioning that Theorem~\ref{Result:RTWitness} answers the question~\ref{Q:Characterisation}.
As we shall demonstrate in the later section, questions~\ref{Q:Comparison} and~\ref{Q:Quantification} can also be answered as desired for $R$-free incompatibility.

\subsection{Robustness Measure as a Conic Program and the Minimal Assumptions for Its Regularity}
In Appendix~\ref{App:RConeProgram} we show that $2^{\RRT(\bmsl)}$ is the solution of the following optimization problem:
\begin{subequations}\label{Eq:RConicProgram}
\begin{eqnarray}\label{Eq:trV}
\begin{aligned}
	\min_{V}\quad&\tr(V)\\
	{\rm s.t.}\quad& V\in\mathbfcal{C}_{R|\rmt};\;\sigma_{\rm X}\preceq\tr_{\rm S\setminus X}(V)\quad\forall\;{\rm X\in}{\bm{\Lambda}},
\end{aligned}
\end{eqnarray}
where 
\begin{align}\label{Eq:C-Cone}
	\mathbfcal{C}_{R|\rmt}\coloneqq\{\alpha\eta_\rms\,|\,\alpha\ge0,\eta_\rms:{\rm state\;s.t.}\;\tr_{\rm S\setminus T}(\eta_\rms)\in\FRT\}.
\end{align}
\end{subequations}
Thus, as long as $\mathbfcal{C}_{R|\rmt}$ is a cone, we may cast the computation of $\RRT$ as a conic program, cf.~\cref{Eq:CPPrimal}. 
However, even then, strong duality may not be guaranteed to hold, or $\RRT$ could be unbounded. Next, we present the {\em minimal} assumptions required to meet these desirable features.

In Appendix~\ref{App:Proof-Result:Equivalence} we prove the following result:
\begin{proposition}\label{Result:Equivalence}
{\em\bf (Minimal Assumptions)}
Given $R$, $\rmt$, and $\bm{\Lambda}$, the following statements are equivalent:
\begin{enumerate}
\item\label{Condition1} Assumptions~\ref{Assumptions} hold.
\item\label{Condition2} (i) $\mathbfcal{C}_{R|\rmt}$ is a proper cone. (ii) For every $\bmsl$, Eq.~\eqref{Eq:RConicProgram} is a finite, feasible conic program with strong duality.
\end{enumerate}
\end{proposition}
Note that Statement~\ref{Condition2} guarantees the feasibility of \cref{Eq:RConicProgram}, and hence \cref{Eq:Robustness} as an optimization problem. Moreover, Proposition~\ref{Result:Equivalence} illustrates that Assumptions~\ref{Assumptions} are both necessary and sufficient for Eq.~\eqref{Eq:RConicProgram} [and hence \cref{Eq:Robustness}] to be a finite conic program with strong duality.
In particular, Assumptions~\ref{Assumptions} are shared by several common state resource theories (see discussions below), thus making evident the generality of our approach. 

Let us now comment on the significance of these assumptions. 
The convexity of $\FRT$ in Assumption~\ref{Assumption:Convex} implies that probabilistic mixtures of free states are again free.
The compactness of $\FRT$, on the other hand, means that when a state $\rho$ can be approximated by free states with arbitrary precision, then $\rho$ is also free. 
These features are shared by many resources, such as entanglement, coherence, athermality, asymmetry, steerability, and nonlocality (see Appendix~\ref{App:footnote4}).

However, it is important to remark that not every state resource theory satisfies Assumption~\ref{Assumption:Convex}.
First, convexity implicitly allows shared randomness for free.
This is no longer true, for instance, when $\mathcal{F}_R$ is the set of multipartite states that are separable with respect to {\em some} bipartition.
Since a non-trivial convex combination of two states separable in different bipartitions will generally {\em not} result in a state separable with respect to {\em any} bipartition, $\mathcal{F}_R$, and hence $\FRT$ can be non-convex.
Likewise, if one identifies all pure states as the resource, then $\mathcal{F}_R$ will be the set of all non-pure states, which is not closed.

Assumption~\ref{Assumption:Static} implies that there exists a free state $\etaT$ in $\rmt$ that may be extended to $\rms$ as $\rhos$ such that all the corresponding $\rhox$ are full-rank. 
By considering an extension of the kind $\rhos = \frac{\id_{S\setminus T}}{d_{S\setminus T}}\otimes\etaT$, we see that Assumption~\ref{Assumption:Static} holds whenever the following sufficient condition holds:
\begin{enumerate}[label = (A2*)]
\item\label{Asssumption:SufficientCondition} There exists a full-rank $\eta_\rmt\in\FRT$.
\end{enumerate}
This is, however, not necessary, and a counterexample is given in Appendix~\ref{App:RemarkAssumption2}.
From here, we learn that Assumption~\ref{Assumption:Static} is satisfied when a maximally mixed state is a free state.
Hence, all but athermality of the above-mentioned resources fulfill this assumption.
In the case of athermality, when the given thermal state is full-rank, then this property is again satisfied.
This also provides an example where Assumption~\ref{Assumption:Static} becomes invalid, e.g., when the given thermal state is a {\em product} pure state, which can be understood as the zero temperature limit without entanglement.
We also learn that:
\begin{corollary}
If Assumptions~\ref{Assumption:Convex} and~\ref{Asssumption:SufficientCondition} hold for  a given state resource $R$, then $\mathbfcal{C}_{R|\rmt}$ is a proper cone. 
Also, for every $\bmsl$, Eq.~\eqref{Eq:RConicProgram} is a finite conic program with strong duality.
\end{corollary}

\subsection{Operational Interpretation of $R$-Free Incompatibility}\label{Sec:OperationalMeaning}

Interestingly,~\cref{Result:RTWitness} can be used to give an operational interpretation showing how $R$-free incompatibility leads to an advantage in {\em channel discrimination tasks} analogous to those discussed in  Refs.~\cite{Takagi2019-2,Takagi2019}. 
For any given $\bmsl$, let $\EIL\coloneqq\{\mE_{i|{\rm X}}\}_{i,{\rm X}}$ be the set of channels to be distinguished, then the task consists in the following steps:
\begin{enumerate}
\item With probability $p_\rX$, the agent at $\rX\in\bml$ is chosen.
\item With probability $p_{i|{\rm X}}$, the channel $\mE_{i|{\rm X}}\in\EIL$ is chosen to be implemented at $\rX$.
\item The agent at ${\rm X}$ inputs the quantum state $\sigma_\rX$ through the channel and applies the measurement associated with the {\em positive operator-valued measures} (POVMs)~\cite{QCI-book} $\{E_{i|{\rm X}}\}_i$ ($\sum_iE_{i|{\rm X}} = \id_{\rm X}$ and $E_{i|{\rm X}}\succeq 0$) to the channel's output.
\item If the measurement outcome is $j$, the agent guesses $\mE_{j|{\rm X}}$ as the channel implemented. 
\end{enumerate}
See also Fig.~\ref{Fig:ChannelD} for a schematic illustration of this task.
Together, $D\coloneqq \left(\{p_{\rm X}\}_{\rm X},\{p_{i|{\rm X}}\}_{i,{\rm X}},\{E_{i|{\rm X}}\}_{i,{\rm X}}\right)$ and $\EIL$ define the discrimination task. 
For any chosen input states $\bmsl=\{\sigma_\rX\}_{{\rm X}\in\bml}$, the probability of successfully distinguishing members of $\EIL$ in this task, when averaged over multiple rounds, is evidently
\begin{align}\label{Eq:P_D}
	P_D(\bmsl,\EIL)\coloneqq\sum_{{\rm X\in}\bm{\Lambda}}\sum_{i}p_{\rm X}p_{i|{\rm X}}\tr\left[E_{i|{\rm X}}\mE_{i|{\rm X}}(\sigma_{\rm X})\right].
\end{align}
Hereafter, we focus on tasks $D$ that are {\em strictly positive}, i.e.,  $p_{\rm X},p_{i|{\rm X}}>0,E_{i|{\rm X}}\succ0\;\forall\;i\;\&\;{\rm X}$ and where the ensemble of channels is unitary, i.e.,  $\EIL=\mathbfcal{U}:=\{\mathcal{U}_{i|{\rm X}}\}_{i,\rX}$ and each $\mathcal{U}_{i|{\rm X}}$ is unitary.
Then in Appendix~\ref{App:Proof-Result:DiscriminationTask} we show the following:

\begin{theorem}\label{Result:DiscriminationTask}
{\em\bf (Advantage in Channel Discrimination)}
If $R$, {\rm T}, and $\bm{\Lambda}$ satisfy Assumptions~\ref{Assumptions}, then $\bmsl\notin\CRTL$ if and only if for every unitary $\mathbfcal{U} = \{\mathcal{U}_{i|{\rm X}}\}_{i=1;{\rm X\in\bm{\Lambda}}}^{d_{\rm X}+1}$ there exists a strictly positive channel-discrimination task $D$ such that
$\sup_{\bm{\tau_\Lambda}\in\CRTL}P_D(\bm{\tau}_{\bm{\Lambda}},\mathbfcal{U})<P_D(\bmsl,\mathbfcal{U}).$
\end{theorem}
Hence, $R$-free incompatibility implies an advantage in distinguishing reversible channels. A few remarks are now in order. Firstly, as the discrimination task is derived from the witness ${\bf W}_{\bm{\Lambda}}$ of Theorem~\ref{Result:RTWitness}, it is specific to the {\em given} $\bmsl$. In general, we should think of the agents somehow having access to $\bmsl$ and, upon chosen, use the corresponding $\sigma_\rX$ to perform the discrimination task. If $\bmsl$ is compatible, then these $\sigma_\rX$'s are simply the reduced states of some global state $\rho_S$. The resourceful nature of $\rho_\rT$ then guarantees an advantage in the aforementioned discrimination task. Note also that an operational advantage of a wide range of resources in the form of a discrimination task has been discussed, e.g., in Refs.~\cite{Takagi2019,Uola2019,Uola2020}.
Here, we show that this advantage extends to $R$-free incompatibility and can be manifested by considering a discrimination task that is strictly positive and that involves {\em only} unitaries.

In fact, \cref{Result:DiscriminationTask} holds even when one considers the unusual situation where $\mathcal{U}_{i|{\rm X}}=\mathcal{U}_{{\rm X}}$, i.e., where these channels to be ``distinguished'' are the same for all $i$. Evidently, in this case, $D$ and $\EIL=\{\mathbfcal{U}_\rmx\}$ do not carry the meaning of a channel discrimination task. Still, one can see it as a kind of {\em communication game} where the player, knowing $D$, tries to maximize the score given in \cref{Eq:P_D} by sending in the best possible, possibly $R$-free incompatible $\bmsl$. The relevance of such a game will become clear when we discuss the convertibility among different $\bmsl$'s in \cref{Sec:Convertibility}.

\begin{figure}
\begin{center}
\scalebox{0.9}{\includegraphics{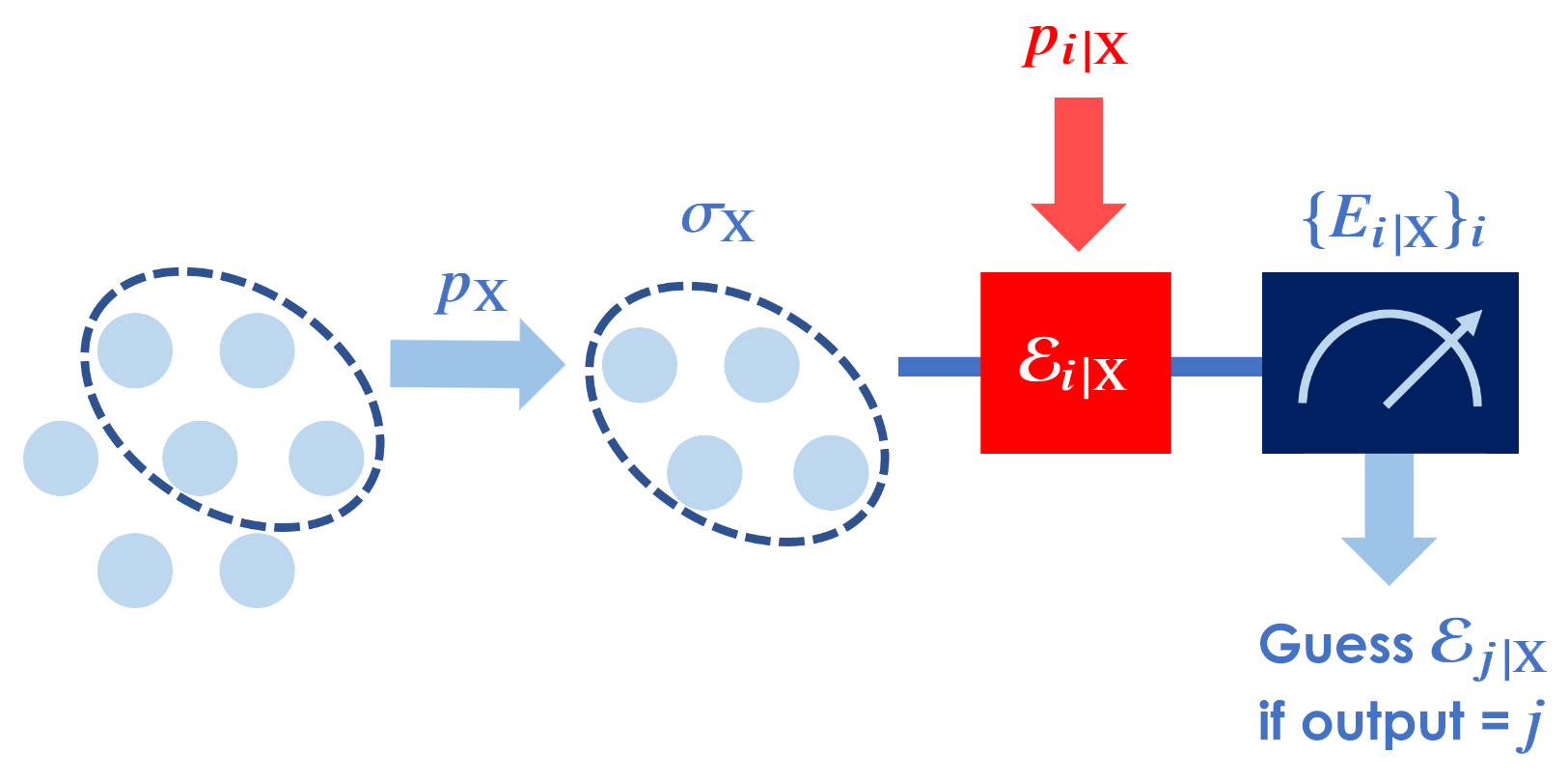}}
\caption{
{\bf Schematic illustration of the channel discrimination task.}
With probability $p_{\rm X}$, the subsystem ${\rm X}\in{\bm\Lambda}$ is selected, and one needs to distinguish channels $\{\mathcal{E}_{i|{\rm X}}\}_i$ implemented with the probability $\{p_{i|{\rm X}}\}_i$.
To do so, the local agent prepares the local state $\sigma_{\rm X}$ as the input of the unknown selected channel, and uses the POVM $\{E_{i|{\rm X}}\}_i$ to measure channel's output.
The agent then guesses the channel is $\mathcal{E}_{j|{\rm X}}$ if the measurement outcome is $j$.
}
\label{Fig:ChannelD} 
\end{center}
\end{figure}

\subsection{Completing the Resource Theory of $R$-Free Incompatibility}
The above observation suggests that $R$-free incompatibility {\em itself} is a resource since it provides advantages in a nontrivial discrimination task that is otherwise absent.
Accordingly, the free quantities are simply members of $\CRTL$. To define the free operations, let us first recall the notion of channel compatibility recently introduced in Ref.~\cite{Hsieh2021}. 
Throughout, we use $\E_{\rm A\to B}$ to represent a channel mapping from systems ${\rm A}$ to ${\rm B}$.
If the input and output Hilbert space are the same, say, $\rmx$, we simply write the map as $\E_\rmx$.
Then, a global channel  $\E_{\rm S'\to S}$ mapping from ${\rm S'}$ to ${\rm S}$ is said to have a {\em well-defined marginal} in the input-output pair ${\rm X'\to X}$  (with ${\rm X'\subseteq S'}$ and ${\rm X\subseteq S}$) if there exists a {\em marginal channel}, denoted by $\E_{\rm X'\to X}$, from ${\rm X'}$ to ${\rm X}$ such that the following commutation relation holds~\cite{Hsieh2021}:
\begin{align}\label{Eq:Tr}
\tr_{\rm S\setminus X}\circ\E_{\rm S'\to S} = \E_{\rm X'\to X}\circ\tr_{\rm S'\setminus X'}.
\end{align}
Once such a marginal exists, it is provably unique~\cite{Hsieh2021}, and we use 
$\ttr_{\rm S'\setminus X'\to S\setminus X}\E_{\rm S'\to S}\coloneqq\E_{\rm X'\to X}$
to denote the marginal channel of $\E_{\rm S'\to S}$ from ${\rm X'\to X}$. Notice that, apart from the capitalization in $\ttr(\cdot)$, the notation for this marginalization operation also differs from the usual partial trace operation $\tr(\cdot)$ in that the subscripts associated with $\ttr$ always take the form of ``$\cdot\to\cdot$''.
For $\E_{\rm S'\to S}$, the existence of $\ttr_{\rm S'\setminus X'\to S\setminus X}\E_{\rm S'\to S}$ is equivalent to being semi-causal in $\rX$~\cite{Beckman2001,Piani2006,Eggeling2002}, semi-localizable in $\rX$~\cite{Eggeling2002}, and no-signaling from ${\rm S'\setminus X'}$ to $\rX$~\cite{Duan2016}.

In these terms, we say that a global channel $\mE_{\rms}$ on ${\rm S}$ is {\em compatible} with a set of channels $\mathbfcal{E}_{\bm{\Lambda}} = \{\mE_{\rm X}\}_{{\rm X\in}\bm{\Lambda}}$ acting on subsystem(s) $\rX$ if $\ttr_{\rm S\setminus X\to S\setminus X}\mE_\rms = \mE_{\rm X}\;\forall\;{{\rm X\in}\bm{\Lambda}}$.
We are now ready to define the free operation of the resource theory associated with $R$-free incompatibility as:
\begin{align}\label{Eq:FreeOperation}
	\mathfrak{O}_{R|\rmt,\bm{\Lambda}}\coloneqq\{\mathbfcal{E}_{\bm{\Lambda}}\;|\;\exists\;\E_\rms\;{\rm compatible\;with}\;\mathbfcal{E}_{\bm{\Lambda}}\;{\rm s.t.}\;\ttr_{\rm 	S\setminus T\to S\setminus T}\E_\rms\in\mathcal{O}_R\},
\end{align}
where $\mathcal{O}_R$ denotes the set of free operations of the given state resource $R$.
Indeed, the legitimacy of this choice follows directly (see \cref{Proof-Eq:GoldenRuleTest} for details) from the definition of $\mathfrak{O}_{R|\rmt,\bm{\Lambda}}$ given in~\cref{Eq:FreeOperation} and that of $\CRTL$ given in \cref{Eq:CTR}: 
\begin{align}\label{Eq:GoldenRuleTest}
	\mathbfcal{E}_{\bm{\Lambda}}(\bm{\tau}_{\bm{\Lambda}})\in\CRTL\;\forall\;\bm{\tau}_{\bm{\Lambda}}\in\CRTL \text{ and }\forall\;	\mathbfcal{E}_{\bm{\Lambda}}\in\mathfrak{O}_{R|\rmt,\bm{\Lambda}}.
\end{align}
Moreover, as we show in \cref{Proof-Result:R-CompatibilityMonotone}, $\RRT$ is a monotone with respect to this choice of free operations.
\begin{theorem}\label{Result:R-CompatibilityMonotone}
{\em\bf ($R$-Free Incompatibility Monotone)}
If $R$, {\rm T}, and $\bm{\Lambda}$ satisfy Assumptions~\ref{Assumptions}, then
\begin{enumerate}
	\item$\RRT(\bm{\tau}_{\bm{\Lambda}}) = 0$ if and only if $\bm{\tau}_{\bm{\Lambda}}\in\CRTL$.
	\item$\RRT\left[\mathbfcal{E}_{\bm{\Lambda}}(\bmsl)\right]\le\RRT(\bmsl)$ $\forall\;\bmsl,\forall\;\mathbfcal{E}_{\bm{\Lambda}}	\in\mathfrak{O}_{R|\rmt,\bm{\Lambda}}$.
\end{enumerate}
\end{theorem}
Hence, not only can $\RRT$ quantify $R$-free incompatibility, but it is also a monotone for the judicious choice of free operations given in \cref{Eq:FreeOperation}.
Together with Theorem~\ref{Result:R-CompatibilityMonotone}, we thus complete specifying the resource theory associated with $R$-free incompatibility.
Note that Theorem~\ref{Result:R-CompatibilityMonotone} also helps us to answer question~\ref{Q:Quantification} for $R$-free incompatibility.

\subsection{Convertibility Conditions}\label{Sec:Convertibility}
As mentioned in~\cref{Sec:QRT}, a key question in the study of resource theory is to characterize the convertibility of resources under free operations [question~\ref{Q:Comparison}].
It turns out that we can achieve this by using a specific kind of communication game, c.f.,~\cref{Sec:OperationalMeaning}.
For a given resource $R$ and a game specified by $D\coloneqq \left(\{p_{\rm X}\}_{\rm X},\{p_{i|{\rm X}}\}_{i,{\rm X}},\{E_{i|{\rm X}}\}_{i,{\rm X}}\right)$, we consider the following figure-of-merit:
\begin{align}
\Scr(\bmsl,D)\coloneqq\sup_{\mathbfcal{E}_{\bm{\Lambda}}\in\mathfrak{O}_{R|\rmt,\bm{\Lambda}}}P_D\left(\bmsl,\mathbfcal{E}_{\bm{\Lambda}}\right),
\end{align}
where $P_D\left(\bmsl,\mathbfcal{E}_{\bm{\Lambda}}\right)$ has been defined in Eq.~\eqref{Eq:P_D}, and $\mathfrak{O}_{R|\rmt,\bm{\Lambda}}$ is the set of free operations associated with the resource theory of $R$-free incompatibility, as defined in Eq.~\eqref{Eq:FreeOperation}.
In other words, ``$\Scr$'' is the highest ``success probability'' obtained by optimizing over all channels coming from the free operations of $R$-free incompatibility.
Note that $\Scr(\bmsl,D)$ is also dependent on $\mathfrak{O}_{R|\rmt,\bm{\Lambda}}$ but we keep this dependency implicit.

Unlike in~\cref{Sec:OperationalMeaning}, we now allow subnormalized probabilities such that $\sum_{\rm X}p_{\rm X}\le1$ and $\sum_i p_{i|{\rm X}}\le1$. 
We call such tasks $D$ {\em non-deterministic} communication games.
Then we show in~\cref{App:Proof-Result:Convertibility} that non-deterministic communication games can fully characterize the convertibility of $R$-free incompatibility:
\begin{theorem}\label{Result:Convertibility}
{\em\bf(Convertibility of $R$-Free Incompatibility)}
Suppose $\mathfrak{O}_{R|\rmt,\bm{\Lambda}}$ has the following three properties: 
\begin{itemize}
	\item It is convex and compact.
	\item $\mathbfcal{E}_{\bm{\Lambda}}'\circ\mathbfcal{E}_{\bm{\Lambda}}\in\mathfrak{O}_{R|\rmt,\bm{\Lambda}}$ if $\mathbfcal{E}_{\bm{\Lambda}}',\mathbfcal{E}_{\bm{\Lambda}}\in\mathfrak{O}_{R|\rmt,\bm{\Lambda}}$.
	\item Identity is a free operation.
\end{itemize}
Then the following two statements are equivalent:
\begin{enumerate}
	\item\label{Free_op_transformation} There exists $\mathbfcal{E}_{\bm{\Lambda}}\in\mathfrak{O}_{R|\rmt,\bm{\Lambda}}$ such that $\mathbfcal{E}_{\bm{\Lambda}}(\bm{\tau}_{\bm{\Lambda}}) = \bm{\sigma}_{\bm{\Lambda}}$.
	\item\label{Psucc_allD} $\Scr({\bm\tau}_{\bm\Lambda},D)\ge \Scr(\bmsl,D)$ for every non-deterministic communication game $D$.
\end{enumerate}
\end{theorem}
Theorem~\ref{Result:Convertibility} provides a complete characterization of the partial ordering in the resource theory of $R$-free incompatibility; that is, being able to convert an $R$-free incompatible $\bm{\tau}_{\bm{\Lambda}}$ to $\bmsl$ is equivalent to the former achieving $\Scr$ at least as high as the latter in {\em all} non-deterministic communication games.

As a summary, we have provided a general and thorough analysis of the resource-theoretic aspects of $R$-free incompatibility.
Especially, we have fully addressed the three important questions mentioned at the beginning: questions~\ref{Q:Characterisation},~\ref{Q:Comparison}, and~\ref{Q:Quantification} are addressed by Theorem~\ref{Result:RTWitness}, Theorem~\ref{Result:R-CompatibilityMonotone}, and Theorem~\ref{Result:Convertibility}, respectively.

\section{Applications and Implications}\label{Sec:Applications}

\subsection{Characterizing Resource Marginal Problems by Many-Body Physics}
As mentioned previously, our general framework does not require $\bmsl\in\CL$ a priori. 
To this end, we say that a resource marginal problem has {\em non-trivial} solutions if $\CL\setminus\CRTL$ is {\em non-empty}. 
In other words, a non-trivial solution involves a $\bmsl$ that is compatible and can demonstrate $R$-free incompatibility.
Clearly, for any given $R$, $\rmt$, and ${\bm\Lambda}$, understanding when non-trivial solutions exist is both physically relevant and important.
Thanks to Theorem~\ref{Result:RTWitness}, we can achieve this by analyzing the ground state behavior in a many-body system.

We argue as follows.
Let ${\bm\Lambda},{\rm T}$ be given, and suppose $R$ satisfies Assumptions~\ref{Assumptions} so that Theorem~\ref{Result:RTWitness} is applicable.
Let us further suppose that there exists non-trivial examples of $R$-free incompatibility.
From Theorem~\ref{Result:RTWitness}, we can find operators $\{W_{\rm X}\}_{{\rm X}\in{\bm\Lambda}}$ such that
\begin{align}
\sum_{{\rm X}\in{\bm\Lambda}}{\rm tr}(\sigma_{\rm X}W_{\rm X})<\inf_{\bm{\tau}_{\bm{\Lambda}}\in\CRTL}\sum_{{\rm X}\in{\bm\Lambda}}{\rm tr}(\tau_{\rm X}W_{\rm X}).
\end{align}
Consider now the (dimensionless) Hamiltonian 
\begin{align}
H_{\rm total}\coloneqq\sum_{{\rm X}\in{\bm\Lambda}}W_{\rm X}\otimes\id_{{\rm S\setminus X}}
\end{align}
and let $\eta_{\rm S}$ be a ground state of this Hamiltonian.
Then we {\em must} have 
\begin{align}\label{Eq:Comp000001}
{\rm tr}\left(\eta_{\rm S}H_{\rm total}\right)\le{\rm tr}\left(\rho_{\rm S}H_{\rm total}\right)<\inf_{\bm{\tau}_{\bm{\Lambda}}\in\CRTL}\sum_{{\rm X}\in{\bm\Lambda}}{\rm tr}(\tau_{\rm X}W_{\rm X}),
\end{align}
where $\rho_{\rm S}$ is any state compatible with $\bm{\sigma}_{\bm{\Lambda}}$ (and, importantly, there exists at least one such state since $\bm{\sigma}_{\bm{\Lambda}}$ is a non-trivial example).
Eq.~\eqref{Eq:Comp000001} implies that {\em all} ground states of $H_{\rm total}$ must have their marginals demonstrating $R$-free incompatibility.
In other words, 
{\em 
every state from the ground energy subspace of $H_{\rm total}$ must be $R$-resourceful in $T$.
}

Conversely, suppose that {\em all} ground states of a given Hamiltonian $H_{\rm total} = \sum_{{\rm X}\in{\bm\Lambda}}H_{\rm X}\otimes\id_{{\rm S\setminus X}}$ are $R$-resourceful in $T$.
Then, we want to know whether this implies any non-trivial example of $R$-free incompatibility for any given $R$, T, and $\bm{\Lambda}$.
Let us pick a ground state $\eta_{\rm S}$ and compute its marginal states $\bm{\sigma}_{\bm{\Lambda}}=\{\sigma_{\rm X}\coloneqq\eta_{\rm X}\}_{{\rm X}\in{\bm\Lambda}}$.
Then this set is, by definition, compatible.
Moreover, for every $\rho_{\rm S}$ compatible with $\bm{\sigma}_{\bm{\Lambda}}$, we have
\begin{align}
	{\rm tr}\left(\rho_{\rm S}H_{\rm total}\right) = \sum_{{\rm X}\in{\bm\Lambda}}{\rm tr}\left(\sigma_{\rm X}H_{\rm X}\right) = {\rm tr}\left(\eta_{\rm S}H_{\rm total}\right) =\text{ground energy}.
\end{align}
Hence, $\rho_{\rm S}$ is also a ground state, which, by assumption, {\em must} be $R$-resourceful in $T$.
Consequently, $\bm{\sigma}_{\bm{\Lambda}}$ is a non-trivial example of $R$-free incompatibility.
Therefore, if there exists a many-body Hamiltonian $H_{\rm total} = \sum_{{\rm X}\in{\bm\Lambda}}H_{\rm X}\otimes\id_{{\rm S\setminus X}}$ whose ground states are all $R$-resourceful in some subsystem $T$, then 
{\em 
there must exist at least one non-trivial example of $R$-free incompatibility.
}

What we have argued can be formally summarized as the following theorem:
\begin{theorem}\label{Result:Htotal}
{\em\bf(Hamiltonian Interpretation)}
Let ${\bm\Lambda}$ and ${\rm T}$ be the given set of subsystems and target system.
Suppose $R$, {\rm T}, and $\bm{\Lambda}$ satisfy Assumptions~\ref{Assumptions}.
Then the following two statements are equivalent:
\begin{itemize}
\item There exists at least one non-trivial example of $R$-free incompatibility. 
\item There exists a many-body Hamiltonian \mbox{$H_{\rm total} = \sum_{{\rm X}\in{\bm\Lambda}}H_{\rm X}\otimes\id_{{\rm S\setminus X}}$} such that all its ground states are $R$-resourceful in $T$.
\end{itemize}
\end{theorem}

This link provides a physical interpretation of the abstract notion of resource marginal problems in a many-body setting.
Namely, if $R$-free incompatibility is non-trivial, then we can design a Hamiltonian whose ground state is resourceful in T, even if the ground state energy is degenerate. In that case, a randomly chosen ground state is guaranteed to demonstrate $R$ in the target system.
For instance, {\em all} non-trivial solutions of entanglement transitivity problems~\cite{Tabia} give rise to Hamiltonians whose ground states must have an entangled marginal in the desired subsystem.

Note that the same argument used to prove Theorem~\ref{Result:Htotal} can also show the following result:
\begin{corollary}\label{Result:Htotal-coro}
Let ${\bm\Lambda}$ and ${\rm T}$ be the given set of subsystems and target system.
Suppose $R$, {\rm T}, and $\bm{\Lambda}$ satisfy Assumptions~\ref{Assumptions}.
Then the following two statements are equivalent:
\begin{itemize}
	\item There exists at least one non-trivial example of $R$-free incompatibility. 
	\item There exists a many-body Hamiltonian $H_{\rm total} = \sum_{{\rm X}\in{\bm\Lambda}}H_{\rm X}\otimes\id_{{\rm S\setminus X}}$ such that all its highest energy excited states are $R$-resourceful in $T$.
\end{itemize}
\end{corollary}
Hence, interestingly, non-trivial $R$-free incompatibility implies the existence of Hamiltonians whose highest energy excited states are useful in {\em two} ways: they are the energy eigenstates carrying the highest amount of energy in the system, and they must have the resource $R$ in a given target (sub)system $T$.
Hence, 
{\em
non-trivial $R$-free incompatibility can be translated into a many-body resource that carries energy and quantum information resources simultaneously.
}

\subsection{Applications to State Resource Theories and Marginal Problems}
We now illustrate the versatility of our framework by considering several explicit examples.
By choosing ${\rm T = S}$ and $\bm{\Lambda} = \{\rms\}$, $\CRTL$ and $\mathfrak{O}_{R|\rmt,\bm{\Lambda}}$ reduce, respectively, to  $\mathcal{F}_R$ (more precisely, $\FRS$) and $\mathcal{O}_R$. 
Then, the notion of ``$R$-free incompatibility'' is exactly the requirement of being $R$-free. Hence, as we illustrate in~\cref{Fig:Summary}, the resource theory of $R$-free incompatibility reduces to the resource theory of $R$, and the robustness measure $\RRT$ becomes the one induced by the max-relative entropy~\cite{Datta2009}.
Since these observations hold regardless of the choice of S, we obtain the following application of Theorem~\ref{Result:DiscriminationTask}:
\begin{corollary}\label{Cor:RT}
Let $\mathcal{F}_{R}$ be convex and compact, and $d$ be the dimension of the state space of interest.
Suppose there exists a full-rank free state.
Then $\rho\notin\mathcal{F}_{R}$ if and only if for every unitary $\mathbfcal{U} = \{\mathcal{U}_{i}\}_{i=1}^{d+1}$ there exists a strictly positive $D$ such that $\max_{\eta\in\mathcal{F}_R}P_D(\eta,\mathbfcal{U})<P_D(\rho,\mathbfcal{U}).$
\end{corollary}
In contrast with the results previously derived in Refs.~\cite{Takagi2019-2,Takagi2019}, Corollary~\ref{Cor:RT} dictates that an operational advantage in distinguishing channels even holds for {\em every} combination of {\em unitary} channels.

For ${\rm T = S}$ and $\FRS$ being the set of all states on $\rms$, we have $\CRTL = \CL$, and our resource marginal problem becomes the usual quantum state marginal problem for a given $\bmsl$.
Since all states are ``free,'' the requirement of ``$R$-free compatibility'' is simply the requirement of compatible marginal states.
Then, $R$-free incompatibility reduces to the usual marginal state incompatibility\footnote{When ${\rm T = S}$ and $\FRS = $ the set of all states on $\rms$, we have that $\bm{\tau_\Lambda}$ is incompatible if and only if it is $R$-free incompatible.
To show this, suppose the contrary, namely, there exists an $R$-free incompatible $\bm{\sigma_\Lambda}$ that is compatible.
Then there exists a state $\rhos$ compatible with ${\bm{\sigma_\Lambda}}$.
However, according to the definition, we must have $\rhos\notin\FRS$, i.e., it cannot be a state.
This results in a contradiction and hence shows the desired claim.},
and the resource theory of $R$-free incompatibility becomes the resource theory of state incompatibility (see~\cref{Fig:Summary}). Accordingly, $\mathcal{O}_R$ is  the set of all channels on $\rms$, $\CRTL = \CL$ is the set of all compatible $\bm{\tau}_{\bm{\Lambda}}$, and $\mathfrak{O}_{R|\rmt,\bm{\Lambda}}$ is the set of all compatible channels $\mathbfcal{E}_{\bm{\Lambda}}$ acting on $\rX\in\bml$~\cite{Hsieh2021}.
Moreover, Assumptions~\ref{Assumption:Convex} and~\ref{Assumption:Static} are easily verified, thereby giving the following corollary from Theorem~\ref{Result:DiscriminationTask}:
\begin{corollary}
$\bmsl$ is incompatible if and only if for every unitary $\mathbfcal{U} = \{\mathcal{U}_{i|{\rm X}}\}_{i=1;{\rm X\in\bm{\Lambda}}}^{d_{\rm X}+1}$ there exists a strictly positive $D$ such that
$\max_{\bm{\tau}_{\bm{\Lambda}}\in\CL}P_D(\bm{\tau}_{\bm{\Lambda}},\mathbfcal{U})<P_D(\bmsl,\mathbfcal{U}).$
\end{corollary}
This can be seen as an approach alternative to that provided in Ref.~\cite{Hall2007} for witnessing the incompatibility of a given $\bmsl$.

\subsection{Application: Resource Theory Associated with the Entanglement Marginal Problems}
As a third application, consider the case of ${\rm T = S}$ and $\Free = $ the set of fully separable states in some given multipartite system $\rms$. For $\bmsl\in\CL$, this gives exactly the {\em entanglement marginal problem} recently proposed in Ref.~\cite{Navascues2021}, which aims to characterize when a given set of separable marginal density matrices $\bmsl$ necessarily implies that the multipartite global state is entangled. By definition, all such entanglement-implying $\bmsl$ are $R$-free incompatible in the global system $\rms$. 

By \cref{Result:R-CompatibilityMonotone}, cf.~\cref{Fig:Summary}, this incompatibility therefore gives rise to a resource theory defined by the relevant free operations. Since the set $\Free$ is convex and compact, Assumption~\ref{Assumption:Convex} holds. Evidently, the maximally mixed state is a member of $\Free$, thus Assumption~\ref{Assumption:Static} is satisfied too. Hence, \cref{Result:RTWitness} guarantees that the incompatibility of any given $\bmsl$ can always be certified with the help of a certain witness ${\bf W}_{\bm{\Lambda}}$, which admits an operational interpretation (\cref{Result:DiscriminationTask}). The robustness measure of \cref{Eq:Robustness} can then be used to quantify the resourceful nature of $\bmsl$ within this resource theory of incompatibility.

One may also choose $R = $ genuinely multipartite entanglement, meaning that $\mathcal{F}_R$ is the {\em convex hull} of the union of all biseparable states. Then, as with the original entanglement marginal problems~\cite{Navascues2021},  marginal states that are only compatible with a genuinely multipartite entangled global state can be treated as a resource, and our framework immediately provides the corresponding resource theory, resource monotone, and its operational interpretation in terms of a channel-discrimination task.

\subsection{Application to the Transitivity of Quantum Resources}
As another example of application, note that our framework provides a natural starting point for studying the transitivity problem of {\em any} given state resource $R$.
For simplicity, we illustrate this in a tripartite setting with ${\rm S = ABC}$, $\bm{\Lambda} = \{{\rm AB, BC}\}$, and ${\rm T = AC}$.
Inspired by the work of Ref.~\cite{Coretti2011}, we say that the given resource $R$ is {\em transitive} if there exists compatible $\bmsl = \{\sigma_{\rm AB},\sigma_{\rm BC}\}$ such that $\sigma_{\rm AB},\sigma_{\rm BC}\notin\mathcal{F}_R$ and for every $\rhos$ compatible with $\bmsl$, we have $\rho_{\rm AC}\notin\mathcal{F}_R$. In other words, the transitivity of $R$ can be certified by identifying $\bmsl\in\CL$ such that $\bmsl\not\in\CRTL$.

An in-depth analysis focusing on entanglement transitivity and related problems can be found in Ref.~\cite{Tabia} (see also Refs.~\cite{HTYL,HsiehRT} for a discussion involving other resources). 
Here, we focus on using this specific choice of $\rms$, $\rmt$, and $\bm{\Lambda}$ to illustrate the broad applicability of our framework. In particular, for a resource $R$ such that Assumption~\ref{Assumption:Convex} and Assumption~\ref{Assumption:Static} hold, a resource theory, in view of~\cref{Result:R-CompatibilityMonotone}, can again be defined for marginals $\bmsl$ that exhibit resource transitivity. Likewise, a collection of operators ${\bf W}_{\bm{\Lambda}}$ can be used to witness this fact and to demonstrate an advantage in an operational task. The robustness measure of~\cref{Eq:Robustness} can also be used to quantify the resourcefulness of the given $\bmsl$.

As a concrete example,  consider $\bmsl=\bmsl^W=\{\sigma_{\rm AB},\sigma_{\rm AC}\}$ where $\sigma_{\rm AB}=\sigma_{\rm AC}=\sigma^W$, and $\sigma^W$ is the bipartite marginal of the three-qubit $W$-state~\cite{Dur2000} $\ket{W_{\rm ABC}}\coloneqq\frac{1}{\sqrt{3}}\left(\ket{001} +\ket{010}+\ket{100}\right)_{\rm ABC}$. It is known~\cite{wu2014} that the only three-qubit state compatible with this given $\bmsl$ is the $W$-state itself, and hence the $AC$ marginal must also be $\sigma^W$, thereby showing the transitivity of entanglement.

To illustrate the advantage alluded to in~\cref{Result:DiscriminationTask}, we consider a discrimination task $D^{\mathbfcal{U}}\coloneqq\left(\{p_{\rm X},p_{i|{\rm X}}\},\{E^\mathbfcal{U}_{i|{\rm X}}\}\right)$ with $p_{\rm X} = \frac{1}{2}$, $p_{i|{\rm X}} = \frac{0.99}{4}, i=1,2,3,4$ and $p_{5|{\rm X}} = 0.01$ for both ${\rm X=AB,BC}$,\footnote{The strong bias in $\{p_{i|{\rm X}}\}$ originates from our intention to amplify the advantage.} and the POVM elements $E_{i|{\rm X}}$ specified in Appendix~\ref{App:ExamplePOVMs}.
Then, for $N = 10^5$ sets of five unitary matrices $\mathbfcal{U}=\{ U_{i|{\rm X}} \}_{i=1;{\rm X}\in\bm{\Lambda}}^5$, each randomly generated according to the Haar measure, we compute the operational advantage
$
\Delta P = P_D(\bmsl,\mathbfcal{U}) - \max_{\bm{\tau}_{\bm{\Lambda}}\in\CRTL}P_D(\bm{\tau}_{\bm{\Lambda}},\mathbfcal{U})
$
where $\CRTL$ is the set of $\bm{\tau}_{\bm{\Lambda}}$ giving rise to a separable two-qubit state in AC. 
A histogram of the results obtained, see~\cref{FigUnitary}, clearly demonstrates the said advantage mentioned in \cref{Result:DiscriminationTask}.

\begin{figure}
\begin{center}
\includegraphics[width=0.45\textwidth]{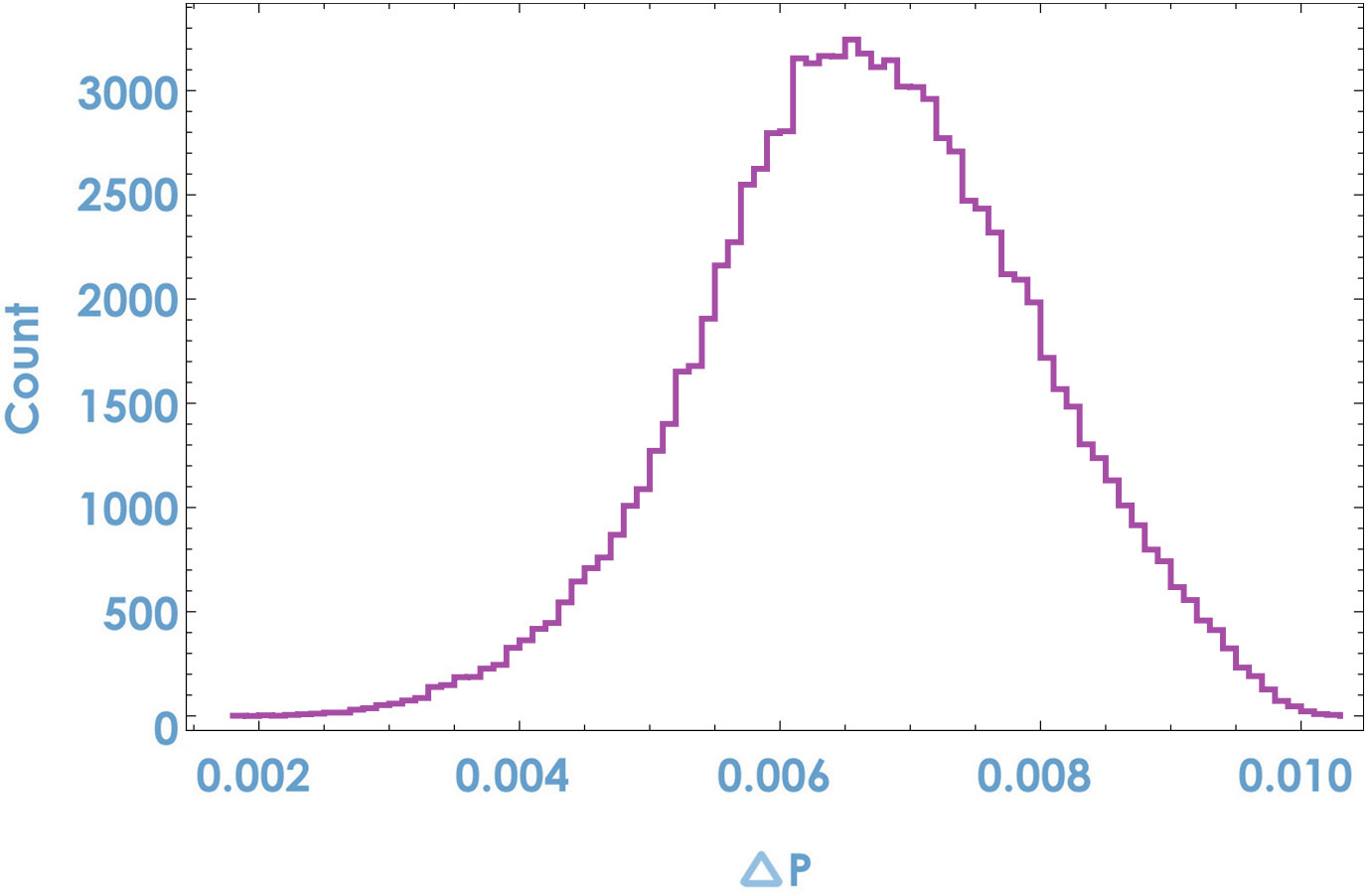}
\caption{
 Histogram of the operational advantage $\Delta P$ derived from $\bmsl^W$ in a series of $N=10^5$ channel-discrimination tasks involving unitary channels randomly chosen according to the Haar measure. 
With bin widths of size $10^{-4}$, our data has the smallest and the largest $\Delta P$ of 
$0.001884$ and $0.010274$, respectively. 
It also has a mean and standard deviation of $0.0066818 \pm 0.0012439$.}
\label{FigUnitary}
\end{center} 
\end{figure}

\section{Discussion}\label{Sec:Conclusion}
Motivated by the importance and the generality of quantum states marginal problems as well as resource theories, we provide an overarching framework that includes not only these two important topics in quantum information theory, but also a variety of other unexplored possibilities. 
The key notion underlying our framework is resource-free ($R$-free) incompatibility in a target subsystem $\rmt$, i.e., the {\em impossibility} of having a resource-free subsystem $\rmt$ for some given set of marginal states $\bmsl$. 
The question of whether this impossibility holds for any given $\bmsl$ leads naturally to what we dub the {\em resource marginal problems}, which includes the quantum state marginal problems as a special case.

To quantify this incompatibility, we introduce a robustness measure $\RRT(\bmsl)$ and show that, provided two necessary and sufficient conditions are satisfied, $\RRT(\bmsl)$ can be evaluated as a finite-valued conic program where strong duality holds. 
Whenever $\bmsl$ is not $R$-free compatible, we demonstrate how a witness can be extracted from the conic program to manifest this fact. Moreover, a channel-discrimination task involving {\em arbitrary} unitary channels can be defined to illustrate the operational advantage of $R$-free incompatible $\bmsl$ over those compatible ones in this task. 
By identifying appropriate free operations, we further prove that a resource theory of $R$-free incompatibility can be formulated with the robustness measure $\RRT(\bmsl)$ serving as the corresponding monotone. 
The corresponding {\em resource theory} for $R$ is then recovered as a special case of our resource theory.
To complete the story, we provide a complete characterization of convertibility of the $R$-free incompatibility of marginals via a communication game.

Apart from recovering the known results, our framework makes evident the fact that an incompatibility problem can be defined for {\em any} resource theory for quantum states, and {\em vice versa}. 
For example, a resource theory can be defined for the usual quantum states marginal problems, an incompatibility problem can be defined for the resource theory of entanglement, and so on. In particular, since a resource theory of entanglement-free incompatibility can be defined for the recently introduced entanglement marginal problems~\cite{Navascues2021}, our robustness measure, etc., can be applied to this incompatibility. 
More generally, if a resource theory or a marginal problem can be cast in a form that fits our framework, the results that we have derived are readily applicable. 
Finally, as an unexpected finding, we report a correspondence between resource marginal problems and ground state properties of many-body Hamiltonians.
More precisely, we have found that the existence of non-trivial solutions to a resource marginal problem is equivalent to the existence of
a many-body Hamiltonian having specific ground state properties.
This gives resource marginal problems a necessary and sufficient description in a physical setting.

Let us conclude by naming some further possibilities for future work. 
First, as is now well known, not only can resource theories be defined for quantum states, but also for quantum channels (see, e.g., Refs.~\cite{Uola2020,Pirandola2017,Hsieh2017,Bu2020,Diaz2018,Rosset2018,Wilde2018,Bauml2019,Seddon2019,LiuWinter2019,LiuYuan2019,
Gour2019-3,Gour2019,Gour2019-2,Gour2020-1,Takagi2019-3,Theurer2019,Saxena2019,Hsieh2020-2,Hsieh2020-3,Hsieh2021,Hsieh2022}). 
By following a treatment very similar to that of this work, we also establish the {\em dynamical} analogs of many of the results mentioned above.  
In the dynamical regime, it would be interesting to see how our framework can be used to obtain further insights into related problems such as channel broadcasting, measurement incompatibility, causal structures, channel extendibility, etc.
We refer the reader to the follow-up paper Ref.~\cite{HsiehDRMP} (see also Ref.~\cite{Hsieh2021,HaapasaloQuantum}), which lies outside the scope of the present work.
Also, with our framework, existing results from various works, including those in 
Refs.~\cite{Hall2007,Takagi2019-2,Takagi2019,Uola2019,Uola2020} can be recovered.
But given the versatility of resource marginal problems, it should be clear that there remain many other possibilities that are worth exploring beyond those explicitly discussed here.
For instance, it is rewarding to study whether $R$-free incompatibility can be characterized in certain types of exclusion tasks and hence provide advantages in encryption tasks~\cite{Hsieh2023-2} and economy~\cite{Ducuara2022,Ducuara2023,Ducuara2023-2}.

One may also wonder whether it is possible to have some $R$-dependent conditions on each $\sigma_{\rm X}$ that can be used to certify $R$-free incompatibility of $\bmsl$. 
Specifically, could it be that by verifying that all $\sigma_{\rm X}\in\bmsl$ are $R$-resourceful, one can already conclude that $\sigma_\rmt$ is also $R$-resourceful? 
Note that such conditions are {\em generally} neither necessary nor sufficient for demonstrating $R$-free incompatibility.
For instance, with $R$ being entanglement, it is known~\cite{Tabia} that separable $\sigma_{\rm X}$'s can also be constraining enough to force some other marginal $\sigma_\rmt$ to be entangled --- a peculiar phenomenon dubbed {\em meta-transitivity of entanglement} in Ref.~\cite{Tabia}. Hence, such $R$-dependent conditions are generally not necessary for demonstrating $R$-free incompatibility. To see their insufficiency, note that entangled 
$\bmsl$ may not even be compatible with entangled $\sigma_\rmt$. As an example, let $\bmsl=\{\sigma_{\rm AB},\sigma_{\rm BC}\}$ be bipartite marginals obtained from the global state $\proj{\Psi^-}_{\rm AB_1}\otimes\proj{\Psi^-}_{\rm B_2C}$, where ${\rm B=B_1B_2}$ and $\ket{\Psi^-}$ is a Bell state. Then, it is easy to verify that both $\sigma_{\rm X}$ are entangled while the target marginal $\sigma_{\rm AC}$ is separable. Still, to make the certification of $R$-free incompatibility experimental friendly, it will be worth developing some sufficient $R$-dependent conditions for certifying $R$-free incompatibility.

\section*{Acknowledgements}
We thank (in alphabetical order) Antonio Ac\'in, Swati Kumari, Matteo Lostaglio, Shiladitya Mal, Paul Skrzypczyk, and Roope Uola for fruitful discussions and comments. We are also grateful to two anonymous referees for their very helpful remarks and questions on an earlier version of this manuscript.
We acknowledge support from the ICFOstepstone (the Marie Sk\l odowska-Curie Co-fund GA665884), the Spanish MINECO (Severo Ochoa SEV-2015-0522), the Government of Spain (FIS2020-TRANQI and Severo Ochoa CEX2019-000910-S), Fundaci\'o Cellex, Fundaci\'o Mir-Puig, Generalitat de Catalunya (SGR1381 and CERCA Programme), the ERC AdG (on grants CERQUTE and FLQuant), the AXA Chair in Quantum Information Science, the Royal Society through Enhanced Research Expenses (on grant NFQI), the National Science and Technology Council, Taiwan (Grants No.~109-2112-M-006-010-MY3, 112-2628-M-006-007-MY4, 111-2628-E-A49-024-MY2), the National Center for Theoretical Sciences, Taiwan, and Foxconn Quantum Computing Research Center.

\appendix

\section{Preliminary Notions on Topological Properties}\label{App:MathPreliminary}
In this section, we briefly discuss various topological properties of $\CRTL$ and $\mathbfcal{C}_{R|\rmt}$.
To avoid over-complicating this Appendix, we omit proofs of results that are probably well-known to readers with related mathematical backgrounds.
For readers who are not familiar with convex analysis and related ingredients, we refer the readers to Appendix~A of the earlier arXiv version of this paper~\cite{RMPv1}, which contains detailed proofs of every result for pedagogical reasons.

To start with, we clarify the notations.
With a given global system $\rms$ and a set of local systems $\bm{\Lambda}$, we define
\begin{align}\label{Eq:CompactibilitySet}
\CL\coloneqq\{\bmsl\,|\,\exists\;\rhos\;{\rm s.t.}\;\sigma_{\rm X} = \rho_{\rm X}\;\forall\;{\rm X\in}\bm{\Lambda}\},
\end{align}
which is the set of all compatible $\bmsl$.
Here we adopt the notation $\rho_{\rm X}\coloneqq\tr_{\rm S\setminus X}(\rhos)$.
The set of all quantum states on $\rms$ is denoted by
\begin{align}\label{Eq:StateSpace}	
\mathcal{S}_\rms\coloneqq\{\rhos\,|\,\rhos\oge0,\tr(\rhos) = 1\}.
\end{align}
For the convenience of subsequent discussions, we define
\begin{align}\label{Eq:STR}
\mathcal{S}_{\RT}\coloneqq\left\{\rhos\in\mathcal{S}_\rms\,|\,\rho_\rmt\in\mathcal{F}_\RT\right\},
\end{align}
which is a set of global states with free marginals in the target system $\rmt$.
Moreover, for a given set $Q\subseteq\mathcal{S}_\rms$, let the unnormalized versions of $Q$, i.e., the cone corresponding to $Q$, be:
\begin{align}\label{Eq:CQ}
\C_Q\coloneqq\{\alpha\rhos\,|\,\alpha\ge0,\rhos\in Q\}.
\end{align} 
Then, Eq.~\eqref{Eq:C-Cone} can be rewritten as
\begin{align}\label{Eq:CQ-STR}
\mathbfcal{C}_\RT\coloneqq\{\alpha\rhos\,|\,\alpha\ge0,\rhos\in\mathcal{S}_\rms,\rho_\rmt\in\mathcal{F}_\RT\} = \{\alpha\rhos\,|\,\alpha\ge0,\rhos\in\mathcal{S}_{\RT}\} = \C_{\mathcal{S}_\RT}.
\end{align}
In these notations, we have $\mathcal{F}_\RT = \mathcal{F}_R\cap\mathcal{S}_\rmt$.
For convenience, we first list important observations needed for subsequent proofs below.
Recall that a cone $\C$ is {\em pointed} if $x\in\C$ and $-x\in\C$ imply $x = 0$; it is {\em proper} if it is nonempty, convex, compact, and pointed.
\begin{atheorem}\label{Result:PhysicalMessage}
Given a state resource $R$, then
\begin{enumerate}
\item $\mathcal{F}_\RT$ is convex and compact implies that $\CRTL$ is convex and compact.
\item $\mathcal{F}_\RT$ is nonempty, convex, and compact if and only if $\mathbfcal{C}_\RT$ is a proper cone.
\end{enumerate}
\end{atheorem}
The first statement is a combination of Lemma~\ref{Fact:Convex} and Lemma~\ref{Coro:CTR}, and
the second statement is a consequence of Lemma~\ref{Fact:Convex}, Lemma~\ref{Fact:Compact}, and the fact that $\mathbfcal{C}_\RT$ is by definition pointed; namely, when we have an element $x\in\mathbfcal{C}_\RT$ such that $-x\in\mathbfcal{C}_\RT$, then we must have $x = 0$, since we must have $x\ge0$ and $x\le0$ simultaneously.

Next, we note the following facts, whose proof can be found in Ref.~\cite{RMPv1} (throughout, we assume that $\rms,\bm{\Lambda},\rmt$ are given and fixed):

\begin{alemma}\label{Fact:Convex}
The following statements are equivalent:
\begin{enumerate}
\item\label{Conv1} $\mathcal{F}_\RT$ is convex.
\item\label{Conv2} $\CRTL$ is convex.
\item\label{Conv3} $\mathcal{S}_\RT$ is convex.
\item\label{Conv4} $\mathbfcal{C}_\RT$ is convex.
\end{enumerate}
\end{alemma}
To apply Slater's condition~\cite{Boyd:Book}, we need the notion of {\em relative interior}, which we now introduce.
Formally, for a set $Q\subseteq\mathbb{R}^N$ with a given $N\in\mathbb{N}$, its {\em relative interior} is defined by (see, e.g., Refs.~\cite{LectureNote,Bertsekas-book})
\begin{align}\label{Eq:relintDefinition}
{\rm relint}(Q)\coloneqq\left\{x\in Q\,|\,\exists\,\epsilon>0\;{\rm s.t.}\;\mathcal{B}(x;\epsilon)\cap{\rm aff}(Q)\subseteq Q\right\},
\end{align}
where $\mathcal{B}(x;\epsilon)\coloneqq\{y\in\mathbb{R}^N\,|\,\norm{x-y}<\epsilon\}$ is an open ball centering at $x$ with radius $\epsilon$ induced by the usual distance for vectors $\norm{\cdot}$ (note that one can also choose it as one norm or sup norm, since they induce the same topology\footnote{\label{footnote3}This can be seen by the fact that $\norm{\sum_{i,j}a_{ij}\ket{i}\bra{j}}_\infty\le\sum_{i,j}|a_{ij}|\norm{\ket{i}\bra{j}}_\infty\le\sum_{i,j}|a_{ij}|\lesssim\norm{\sum_{i,j}a_{ij}\ket{i}\bra{j}}_\infty$, where the second estimate follows from the fact that $|\bra{i}O\ket{j}|\le4\norm{O}_\infty$ for every linear operator $O$ (see, e.g., Fact F.1 in Ref.~\cite{HsiehMasterThesis}).}), and ${\rm aff}(Q)$ is the affine hull of $Q$~\cite{LectureNote,Bertsekas-book}.
When $Q$ is convex, its relative interior is also given by~\cite{LectureNote,Bertsekas-book}
\begin{align}\label{Eq:Relint}
{\rm relint}(Q)=\left\{x\in Q\;\middle|\;\forall\,y\in Q, \exists\,p\in(0,1)\;\&\;\exists\, z\in Q\;{\rm s.t.}\;x = pz + (1-p)y\right\}.
\end{align}
In other words, $x$ can be written as a ``strict'' convex combination of two members in $Q$.
We then have the following fact (see Ref.~\cite{RMPv1} for a proof):
\begin{alemma}\label{Lemma:RelativeInterior}
Let $Q\subseteq\mathcal{S}_\rms$ be a nonempty convex set.
If $\eta_\rms\in{\rm relint}(Q)$, then $\alpha\eta_\rms\in{\rm relint}(\C_Q)$ for every $\alpha>0$.
\end{alemma}

Finally, we briefly discuss compactness.
From now on, the topology defined for states will be understood as the one induced by the trace norm $\norm{\cdot}_1$.
Also, to talk about topology of sets of states $\bm{\sigma}$, we consider the distance measure 
\begin{align}\label{Eq:VectorNorm}
\norm{\bm{\sigma} - \bm{\tau}}\coloneqq\sum_{\rm X\in{\bm\Lambda}}\norm{\sigma_{\rm X} - \tau_{\rm X}}_1.
\end{align}
One can check that this gives a metric since the triangle inequality of Eq.~\eqref{Eq:VectorNorm} directly follows from the triangle inequality of $\norm{\cdot}_1$~\cite{QCI-book}.
In what follows, the topological properties are understood as induced by this norm.

Now we note the following facts, whose proofs follow from standard facts in mathematical analysis:
\begin{alemma}\label{Fact:CompactC}
$\mathfrak{C}_{\bm\Lambda}$, $\mathcal{S}_\rms$, and $\mathcal{S}_{\rm S|pure}$ are compact, where $\mathcal{S}_{\rm S|pure}$ is the set of all pure states in $\rms$.
\end{alemma}
\begin{alemma}\label{Fact:Compact}
The following statements are equivalent:
\begin{enumerate}
\item\label{Comp1} $\FRT$ is compact.
\item\label{Comp2} $\SRT$ is compact.
\item\label{Closed3} $\mathbfcal{C}_{R|{\rm T}}$ is closed.
\end{enumerate}
\end{alemma}
\begin{alemma}\label{Coro:CTR}
$\CRTL$ is compact when $\FRT$ is compact.
\end{alemma}
Hence, the compactness of the set $\CRTL$ is determined by the compactness of $\FRT$, which shows again the role of Assumption~\ref{Assumption:Convex} in our approach.

\subsection{Compactness of $\mathcal{F}_R$: Some Illustrative Examples}\label{App:footnote4}
This section constitutes the proof of the following statement (recall that we always consider the topology induced by the trace norm $\norm{\cdot}_1$): 
\begin{alemma}
For every finite-dimensional system ${\rm S}$, $\FRS$ is compact if $R\in\{$athermality, entanglement, coherence, asymmetry, nonlocality, and steerability$\}$.
\end{alemma}

First, it is straightforward to see the compactness for athermality, since, in this resource theory, there is only one free state, which is the thermal state (for some given Hamiltonian and temperature).
On the other hand, the compactness of the set of separable states in finite dimensions is a well-known fact~\cite{Ent-RMP}
which explains the validity of Assumption~\ref{Assumption:Convex} for entanglement.
Now, we note the following observation: 
\begin{afact}
	For every finite-dimensional ${\rm S}$, $\FRS$ is compact if a continuous resource-destroying map of $R$ exists. 
\end{afact}
To be precise, a {\em resource destroying map}~\cite{Liu2017} $\calL$ of the given state resource $R$ in a system $\rms$ is a (not necessarily linear) map such that $\calL(\eta) = \eta\;\forall\;\eta\in\FRS$ and $\calL(\mathcal{S}_\rms) = \FRS$.
Thus, if $\calL$ is continuous, then $\FRS$ is compact since a continuous function maps a compact set to a compact set (see, e.g., Refs.~\cite{Apostol-book,Munkres-book}).
One can check that this is indeed the case for coherence and asymmetry: For the former, one can use the dephasing channel $(\cdot)\mapsto\sum_{i}\proj{i}\cdot\proj{i}$, and for the latter one can use $G$-twirling operation, where $G$ is the group defining the symmetry.
Hence, Assumption~\ref{Assumption:Convex} is indeed satisfied by coherence and asymmetry.

It remains to address Bell-nonlocality~\cite{Bell-RMP} and steerability, and we focus on the former since the structure of the proof is the same.
We shall use a bipartite system ${\rm AB}$ to illustrate the idea.
Let us first recapitulate the notion of a Bell inequality. In a bipartite system ${\rm AB}$, a probability distribution ${\bf P} = \{P(ab|xy)\}$ is said to be describable by a {\em local-hidden-variable} (LHV) model, denoted by ${\bf P}\in{\rm LHV}$, if $P(ab|xy) = \sum_\lambda p_\lambda P(a|x,\lambda)P(b|y,\lambda)$
for some probability distributions $\{p_\lambda\},\{P(a|x,\lambda)\},\{P(b|y,\lambda)\}$.
A linear {\em Bell inequality} in a constraint satisfied by all ${\bf P}\in{\rm LHV}$ and may be characterized by a vector ${\bf B} = \{B_{ab|xy}\}$ and a real number $\omega({\bf B})$ such that
\begin{align}
\langle{\bf B},{\bf P}\rangle\coloneqq\sum_{a,b,x,y}B_{ab|xy}P(ab|xy)\le\omega({\bf B}),
\end{align}
where each $B_{ab|xy}\in\mathbb{R}$ and $\omega({\bf B})\coloneqq\sup_{{\bf P}\in{\rm LHV}}\langle{\bf B},{\bf P}\rangle$
is the largest value of $\langle{\bf B},{\bf P}\rangle$ achievable by members in ${\rm LHV}$.
A violation of such an inequality certifies the non-classical nature of the given probability distribution ${\bf P}$, and it is an intriguing fact that this can be attained by certain ${\bf P}$ given by quantum theory.

Formally, ${\bf P}$ is called {\em quantum} if there exists a state $\rho_{\rm AB}$ and a set of local POVMs ${\bf E}_{\rm AB}\coloneqq\{E^{a|x}_{\rm A},E^{b|y}_{\rm B}\}$ (i.e., $\sum_aE^{a|x}_{\rm A} = \id_{\rm A}\;\forall\;x$ and $\sum_bE^{b|y}_{\rm B} = \id_{\rm B}\;\forall\;y$) such that
\begin{align}
P(ab|xy) = \tr\left[\left(E^{a|x}_{\rm A}\otimes E^{b|y}_{\rm B}\right)\rho_{\rm AB}\right]\quad\forall\;a,b,x,y.
\end{align}
We write ${\bf P}_{\rho_{\rm AB}|{\bf E}_{\rm AB}}\coloneqq\left\{\tr\left[\left(E^{a|x}_{\rm A}\otimes E^{b|y}_{\rm B}\right)\rho_{\rm AB}\right]\right\}$ to be the probability distribution induced by the state $\rho_{\rm AB}$ and POVMs ${\bf E}_{\rm AB}$.
In these notations, one can define the set of states that cannot violate {\em any} Bell inequality as :
\begin{align}
\calL_{\rm AB}\coloneqq\left\{\eta_{\rm AB}\;|\;\langle{\bf B},{\bf P}_{\eta_{\rm AB}|{\bf E}_{\rm AB}}\rangle\le\omega({\bf B})\;\forall\;{\bf B}\;\&\;{\bf E}_{\rm AB}\right\}.
\end{align}
We call them {\em local} states and states that are not local are called {\em nonlocal}.
Now we can show that:
\begin{alemma}
$\calL_{\rm AB}$ is compact.
\end{alemma}
\begin{proof}
It suffices to show the closedness since it is a subset of $\mathcal{S}_{\rm AB}$ (see Lemma~\ref{Fact:CompactC}). 
Let $\rho_{\rm AB}\in\mathcal{S}_{\rm AB}$ and $\{\eta_{\rm AB}^{(k)}\}_{k=1}^\infty\subseteq\calL_{\rm AB}$ be a sequence of states such that $\lim_{k\to\infty}\norm{\rho_{\rm AB} - \eta_{\rm AB}^{(k)}}_1 = 0$. 
For every Bell inequality specified by ${\bf B}$ and local POVMs ${\bf E}_{\rm AB}$, we may define the (Hermitian) Bell operator~\cite{Braunstein:PRL:1992} as $\calB:=\sum_{a,b,x,y} B_{ab|xy} E_{\rm A}^{a|x}\otimes E_{\rm B}^{b|y}$, then
\begin{align}
	\left\langle{\bf B},{\bf P}_{\rho_{\rm AB}|{\bf E}_{\rm AB}}\right\rangle&= \left\langle{\bf B},{\bf P}_{\rho_{\rm AB}-\eta_{\rm AB}^{(k)}|{\bf E}_{\rm AB}}	\right\rangle  + \left\langle{\bf B},{\bf P}_{\eta_{\rm AB}^{(k)}|{\bf E}_{\rm AB}}\right\rangle\nonumber\\
	&= \tr\left[\calB \left(\rho_{\rm AB}-\eta_{\rm AB}^{(k)}\right)\right] + \left\langle{\bf B},{\bf P}_{\eta_{\rm AB}^{(k)}|{\bf E}_{\rm AB}}\right\rangle\nonumber\\
	&\le \norm{\calB}_\infty\times\norm{\rho_{\rm AB}-\eta_{\rm AB}^{(k)}}_1 + \omega({\bf B}),\label{Eq:CompactBell}
\end{align}
where the second equality follows from the fact that the Bell value $\left\langle{\bf B},{\bf P}_{\rho_{\rm AB}-\eta_{\rm AB}^{(k)}|{\bf E}_{\rm AB}}\right\rangle$ can be written as the Hilbert-Schmidt inner product between the Bell operator $\calB$ and the difference in the density matrices $\rho_{\rm AB}-\eta_{\rm AB}^{(k)}$, 
and the inequality follows from H\"older's inequality [see, e.g., Eq.~(6) in Ref.~\cite{Baumgartner2011} and Section 1.1 in Ref.~\cite{WatrousBook}) 
  and the assumption that $\eta_{\rm AB}^{(k)}\in\calL_{\rm AB}$.
Since this is true for every $k$, we learn that in the limit of $k\to\infty$,  $\langle{\bf B},{\bf P}_{\rho_{\rm AB}|{\bf E}_{\rm AB}}\rangle\le\omega({\bf B})$, which shows that  $\rho_{\rm AB}\in\calL_{\rm AB}$.
\end{proof}
Hence, Assumption~\ref{Assumption:Convex} holds when the underlying state resource is nonlocality.
Note that here, the nonlocality is not specified to a particular Bell inequality.

\section{Remarks on Strong Duality of Conic Program}\label{App:CPIntro}
As mentioned in the main text, the Slater's condition is equivalent to checking whether there exists a point $x\in{\rm relint}\left(\C\right)$ such that $\mathfrak{L}(x) \prec B$.
When the primal problem is convex, Slater's condition guarantees the strong duality~\cite{Boyd:Book}.
However, there exist examples where the primal problem is convex, yet strong duality does not hold. For example, the primal may be infeasible, and the dual is unbounded.
\begin{afact}\label{DualityExample}
There exist examples where strong duality does not hold, even though the primal problem is convex.
\end{afact}
\begin{proof}
For instance, consider $\C = \{\alpha\proj{0}\,|\,\alpha\ge0\}$ in a qubit system.
Then the following conic program
\begin{eqnarray}
\begin{aligned}
	\min_{V}\quad&\tr(V)\\
	{\rm s.t.}\quad&\frac{\id_2}{2}\ole V;\;V\in\{\alpha\proj{0}\,|\,\alpha\ge0\}
\end{aligned}
\end{eqnarray}
is clearly convex but has no feasible $V$.
Indeed, regardless of the value of $\alpha$, the (matrix) inequality constraint can never be satisfied. Meanwhile, its dual reads
\begin{eqnarray}
\begin{aligned}
	\max_{Y}\quad&\frac{\tr(Y)}{2}\\
	{\rm s.t.}\quad&Y\oge0;\;\bra{0}Y\ket{0}\le1,
\end{aligned}
\end{eqnarray}
which is infinite since $\bra{1}Y\ket{1}$ can be arbitrarily large.
\end{proof}

\section{Conic Programming for $R$-Free Incompatibility}\label{App:RConeProgram}

\subsection{Dual Problem of $\RRT$}
Recall that $\rms$ is always assumed to be finite-dimensional.
Then, we start with the following result, which explicitly explains our motivation for imposing Assumptions~\ref{Assumptions}.
\begin{alemma}\label{Fact:FiniteRobustness}
Suppose Assumptions~\ref{Assumptions} hold, then for every $\bmsl = \{\sigma_{\rm X}\}_{\rm X\in{\bm\Lambda}}$, we have
\begin{enumerate}
\item $\RRT(\bmsl)<\infty$.
\item {\em (Slater's Condition)} There exists $V_*\in{\rm relint}(\mathbfcal{C})$ such that $\sigma_{\rm X}\prec\tr_{\rm S\setminus X}(V_*)\;\forall\;{\rm X}\in\bm{\Lambda}.$
\end{enumerate}
\end{alemma}
\begin{proof}
From the form of the constraint given in \cref{Eq:Robustness}, we see that when Assumption~\ref{Assumption:Static} holds, the optimization of Eq.~\eqref{Eq:Robustness} must be feasible with some finite value of the objective function {\em for every $\bmsl$}.
This implies that $\RRT(\bmsl)<\infty$ for every $\bmsl$.
So it suffices to prove Slater's condition. 
Suppose $\eta_\rms\in\SRT$ satisfies that $\eta_{\rm X}$ is full-rank $\forall\;{\rm X}\in\bm{\Lambda}$, which is guaranteed by Assumption~\ref{Assumption:Static}.
Together with Assumption~\ref{Assumption:Convex} and Lemma~\ref{Fact:Convex}, we learn that $\SRT$ is nonempty and convex.
Hence, its relative interior ${\rm relint}(\SRT)$ is also nonempty and convex~\cite{LectureNote,Bertsekas-book}.
Let $\tau_\rms\in{\rm relint}(\SRT)$.
From the definition Eq.~\eqref{Eq:relintDefinition} we learn that there exist $\epsilon>0$ such that
$
\mathcal{B}(\tau_\rms;\epsilon)\cap{\rm aff}(\SRT)\subseteq\SRT.
$
Here, one can choose the open ball with the trace norm $\norm{\cdot}_1$ (see also \cref{footnote3}).
Now, pick $p\in(0,1)$ to be small enough and define $\kappa_\rms = (1-p)\tau_\rms + p\eta_\rms$.

Then, when $p$ is sufficiently small, we have
\begin{itemize}
\item $\kappa_{\rm X}$ is full-rank $\forall\;{\rm X}\in\bm{\Lambda}$.
\item $\norm{\kappa_\rms-\tau_\rms}_1<\frac{\epsilon}{2}$.
\item $\kappa_\rms\in\SRT$.
\end{itemize}
The last condition is due to the convexity of $\SRT$ (Lemma~\ref{Fact:Convex}).
Using the triangle inequality of $\norm{\cdot}_1$~\cite{QCI-book}, we have $\mathcal{B}\left(\kappa_\rms;\frac{\epsilon}{2}\right)\subseteq\mathcal{B}(\tau_\rms;\epsilon)$; namely,
\begin{align}
\mathcal{B}\left(\kappa_\rms;\frac{\epsilon}{2}\right)\cap{\rm aff}(\SRT)\subseteq\SRT.
\end{align}
Thus, from Eq.~\eqref{Eq:relintDefinition} we learn that $\kappa_\rms\in{\rm relint}(\SRT)$.

From Eq.~\eqref{Eq:CQ-STR} we recall that $\mathbfcal{C}_{R|{\rm T}} = \C_{\SRT}$.
Using Lemma~\ref{Lemma:RelativeInterior}, we have 
\begin{align}
\alpha\kappa_\rms\in{\rm relint}(\mathbfcal{C}_{R|{\rm T}})\quad\forall\;\alpha>0.
\end{align}
Let $p_{\rm min}(\kappa_{\rm X})$ denote the smallest eigenvalue of $\kappa_{\rm X}$.
Being a full-rank state in a finite-dimensional system ${\rm X}$, we have $p_{\rm min}(\kappa_{\rm X})>0\;\forall\;{\rm X}\in\bm{\Lambda}$.
By choosing 
\begin{align}
\alpha_{\rm X} > \frac{1}{p_{\rm min}(\kappa_{\rm X})}\ge1, 
\end{align}
which is finite, we have $\alpha_{\rm X}\kappa_{\rm X}\succ\id_{\rm X}$.
Define 
\begin{align}
V_*\coloneqq\left(\max_{{\rm X}\in\bm{\Lambda}}\alpha_{\rm X}\right)\kappa_\rms\in{\rm relint}(\mathbfcal{C}_{R|{\rm T}}),
\end{align} 
then we have
\begin{align}
\sigma_{\rm X}\ole\id_{\rm X}\prec\alpha_{\rm X}\kappa_{\rm X}\ole\tr_{\rm S\setminus X}(V_*)\quad\forall\;\bmsl\;\&\;{{\rm X}\in\bm{\Lambda}}.
\end{align}
This verifies Slater's condition for the (primal) problem given in Eq.~\eqref{Eq:RConicProgram}, completing the proof.
\end{proof}
We now summarize the conic program of $\RRT$, its dual problem, and a sufficient condition of strong duality in the following theorem.

\newpage

\begin{atheorem}\label{Result:ConicProgrammingFacts}
	Given a state resource $R$, a target system $\rmt$ in a finite-dimensional global system $\rms$, and a set of marginal systems in $\bm{\Lambda}$, then
\begin{enumerate}
\item\label{Statement:Primal}  $2^{\RRT(\bmsl)}$ is given by 
\begin{equation}\tag{\ref{Eq:trV}}
\begin{split}
\min_{V}\quad&\tr(V)\\
{\rm s.t.}\quad&V\in\mathbfcal{C}_{R|{\rm T}};\;\sigma_{\rm X}\ole\tr_{\rm S\setminus X}(V)\quad\forall\;{\rm X}\in\bm{\Lambda},
\end{split}
\end{equation}
which is a conic program whenever $\mathbfcal{C}_{R|{\rm T}}$ defined in Eq.~\eqref{Eq:C-Cone} is a proper cone, i.e., whenever Assumption~\ref{Assumption:Convex} holds.
\item\label{Statement:Finite} $\RRT(\bmsl)<\infty$ for every $\bmsl$ if Assumptions~\ref{Assumptions} hold.
\item\label{Statement:Dual} The optimization problem dual to \cref{Eq:RConicProgram} is given by
\begin{eqnarray}\label{Eq:DualRConeProgram}
\begin{aligned}
\max_{\{Y_{\rm X}\}}\quad&\sum_{{\rm X}\in\bm{\Lambda}}\tr(\sigma_{\rm X}Y_{\rm X})\\
{\rm s.t.}\quad&\sum_{{\rm X}\in\bm{\Lambda}}\tr(\tau_{\rm X}Y_{\rm X})\le1\quad\forall\;{\bm{\tau_\Lambda}\in\CRTL;}\\
&Y_{\rm X}\succeq0\quad\forall\;{\rm X}\in\bm{\Lambda}.
\end{aligned}
\end{eqnarray}
\item\label{Statement:StrongDuality} Strong duality holds if Assumptions~\ref{Assumptions} hold.
\end{enumerate}
\end{atheorem}
\begin{proof}
Statement~\ref{Statement:Finite} follows directly from  Lemma~\ref{Fact:FiniteRobustness}.
The proof of the other Statements proceeds as follows.

{\em Proof of Statement~\ref{Statement:Primal}.--}
From Eq.~\eqref{Eq:Robustness}, it follows that
\begin{eqnarray}
\begin{aligned}
	2^{\RRT(\bmsl)} = \min_{\lambda,\eta_\rms}\quad&\lambda\\
	{\rm s.t.}\quad&\sigma_{\rm X}\preceq\lambda\,\tr_{\rms\setminus\rmx}(\eta_{\rms})\quad \forall\;{\rm X}\in\bm{\Lambda};\\
	&\lambda\ge0; \eta_\rms\succeq0;\;\tr(\eta_\rms) = 1;\eta_\rmt\in\FRT.
\end{aligned}
\end{eqnarray}
Let $V = \lambda\eta_\rms$ and recall from Eq.~\eqref{Eq:CQ-STR} the definition of $\mathbfcal{C}_{R|{\rm T}}$, then the optimization problem becomes Eq.~\eqref{Eq:RConicProgram}.
Finally, note that the cone $\mathbfcal{C}_{R|{\rm T}}$ is nonempty, convex, and closed if and only if $\FRT$ is so (Theorem~\ref{Result:PhysicalMessage}; see also Lemmas~\ref{Fact:Convex} and~\ref{Fact:Compact}),
which is guaranteed from Assumption~\ref{Assumption:Convex}. 
Finally, that \cref{Eq:RConicProgram} is in the primal form of a conic program, cf.~\cref{Eq:CPPrimal}, can be seen by rewriting~\eqref{Eq:trV} as
\begin{equation}\label{Eq:Primal2}
\begin{split}
-\max_{V}\quad&-\tr(V)\\
{\rm s.t.}\quad& V\in\mathbfcal{C}_{R|{\rm T}};\;\mathfrak{L}(V)\preceq -\bigoplus_{\rmx\in\bm{\Lambda}}\sigma_{\rm X},
\end{split}
\end{equation}
where 
$\mathfrak{L}(V)\coloneqq - \bigoplus_{{\rm X}\in\bm{\Lambda}}\tr_{\rm S\setminus X}(V)$ and adopting the Hilbert-Schmidt inner product $\langle x,y\rangle = \tr(x^\dagger y)$.

{\em Proof of Statement~\ref{Statement:StrongDuality}.--} \cref{Eq:RConicProgram} is a convex optimization problem, since the objective function is linear in the variable $V$, and the feasible set is convex, as guaranteed by Assumption~\ref{Assumption:Convex} and Fact~\ref{Fact:Compact}. Together with Slater's condition from Lemma~\ref{Fact:FiniteRobustness}, strong duality holds for~\cref{Eq:RConicProgram}.

{\em Proof of Statement~\ref{Statement:Dual}.--}
Using \cref{Eq:Primal2} and applying the dual form of a conic program, cf. \cref{Eq:CPDual}, from~\cref{Eq:CPPrimal}, we arrive at the following dual program:
\begin{eqnarray}\label{Eq:dual0}
\begin{aligned}
-\min_{Y}\quad&\tr\left[ \left(-\bigoplus_{{\rm X}\in\bm{\Lambda}}\sigma_{\rm X}\right)Y\right]\\
{\rm s.t.}\quad&Y\succeq0;\;\langle Y,\mathfrak{L}(Z)\rangle\ge\langle -\id,Z\rangle\;\forall\;Z\in\mathbfcal{C}_{R|{\rm T}}.
\end{aligned}
\end{eqnarray}
Given the block-diagonal form of the operator involved in the objective function and the constraint, without loss of generality, we may write $Y = \bigoplus_{{\rm X}\in\bm{\Lambda}}Y_{\rm X}$ and define $Z_{\rm X}\coloneqq\tr_{\rm S\setminus X}(Z)$:
\begin{eqnarray}
\begin{aligned}
\max_{\{Y_{\rm X}\}_{\rm X}}\quad&\sum_{{\rm X}\in\bm\Lambda}\tr(\sigma_{\rm X}Y_{\rm X})\\
{\rm s.t.}\quad&\sum_{{\rm X}\in\bm\Lambda}\tr(Z_{\rm X}Y_{\rm X})\le\tr(Z)\\
&\forall\;Z\in\{\alpha\rhos\,|\,\alpha\ge0,\rhos\in\mathcal{S}_\rms,\rho_\rmt\in\FRT\};\\
&Y_{\rm X}\succeq0\quad\forall\;{\rm X}\in\bm{\Lambda}.
\end{aligned}
\end{eqnarray}
Note that this optimization equals the one when we only consider $\alpha>0$.
This is because $\alpha = 0$ gives no constraint on $Y$ (more precisely, it gives the constraint ``$\sum_{{\rm X\in}{\bm\Lambda}}0\le0$''), so the maximization must always be constrained by cases with $\alpha>0$.
This means
\begin{eqnarray}
\begin{aligned}
\max_{\{Y_{\rm X}\}_{\rm X}}\quad&\sum_{{\rmx}\in\bm{\Lambda}}\tr(\sigma_{\rm X}Y_{\rm X})\\
{\rm s.t.}\quad&\sum_{\rmx\in\bm{\Lambda}}\tr(\rho_{\rm X}Y_{\rm X})\le1\quad\forall\;\rhos\in\mathcal{S}_\rms,\, \rho_\rmt\in\FRT;\\
&Y_{\rm X}\succeq0\quad\forall\;\rmx\in\bm{\Lambda}.
\end{aligned}
\end{eqnarray}
This is equivalent to Eq.~\eqref{Eq:DualRConeProgram}, and the proof is completed.
\end{proof}

\section{Other Proofs}
\subsection{Proof of Theorem~\ref{Result:RTWitness}}\label{App:Proof-Result:RTWitness}
\begin{proof}
Let $\bmsl\notin\mathfrak{C}_{R|{\rm T}}$.
In Theorem~\ref{Result:ConicProgrammingFacts}, the dual problem of Eq.~\eqref{Eq:RConicProgram} is shown to be Eq.~\eqref{Eq:DualRConeProgram}.
Furthermore, with the validity of Assumptions~\ref{Assumptions}, Theorem~\ref{Result:ConicProgrammingFacts} implies that the optimization Eq.~\eqref{Eq:RConicProgram} outputs a solution that is finite and strictly larger than $1$.
In other words, there exist $Y_{\rm X}\succeq0, \rmx\in\bm{\Lambda}$ such that
$\sum_{{\rm X}\in\bm{\Lambda}}\tr(\tau_{\rm X}Y_{\rm X})<\sum_{{\rm X}\in\bm{\Lambda}}\tr(\sigma_{\rm X}Y_{\rm X})\;\forall\;\bm{\tau_\Lambda}\in\mathfrak{C}_{R|{\rm T}}.$
Taking $W_{\rmx} = Y_{\rmx}$ completes one direction of the proof. The proof in the other direction is obvious and is thus omitted.
\end{proof}

\subsection{Proof of Proposition~\ref{Result:Equivalence}}\label{App:Proof-Result:Equivalence}
\begin{proof}
First, we note that Assumption~\ref{Assumption:Convex} is necessary and sufficient for the convexity and closedness of $\mathbfcal{C}_{R|{\rm T}}$ (Theorem~\ref{Result:PhysicalMessage}).
Hence, Statement~\ref{Condition2}(i) $\Leftrightarrow$ Assumption~\ref{Assumption:Convex}.
Now suppose $\RRT(\bmsl)<\infty\;\forall\;\bmsl$ {\em but} Assumption~\ref{Assumption:Static} fails, namely, there is no state $\rhos$ in $\rms$ such that $\rho_\rmt\in\FRT$ and $\rho_{\rm X}$ is full-rank $\forall\;{\rm X}\in\bm{\Lambda}$.
This means that for every $\eta_\rms\in\SRT$, we must have that $\eta_{\rm X}$ is {\em not} full-rank for some ${\rm X}\in\bm{\Lambda}$.
Now we choose $0\prec\sigma_{\rm X} = \frac{\id_{\rm X}}{d_{\rm X}}\;\forall\;{\rm X}\in\bm{\Lambda}$.
Then for every $\eta_\rms\in\SRT$, there exists no finite $\alpha>0$ that can achieve
$
\frac{\id_{\rm X}}{d_{\rm X}}\ole\alpha\eta_{\rm X}\;\forall\;{\rm X}\in\bm{\Lambda},
$
since there must be some ${\rm X}$ where $\eta_{\rm X}$ is not full-rank, thereby forbidding the inequality with a finite $\alpha$.
In other words, there is no $V\in\mathbfcal{C}_{R|{\rm T}} = \C_{\SRT}$ that can achieve
$
\frac{\id_{\rm X}}{d_{\rm X}}\ole\tr_{\rm S\setminus X}(V)\;\forall\;{\rm X}\in\bm{\Lambda}.
$
This implies that the (primal) problem of Eq.~\eqref{Eq:RConicProgram} has no feasible point, which gives a contradiction.
Hence, Statement~\ref{Condition2}(ii) $\implies$ Assumption~\ref{Assumption:Static}.

To show the converse, recall from Theorem~\ref{Result:ConicProgrammingFacts} that when Assumptions~\ref{Assumptions} hold, $\RRT(\bmsl)$ for every $\bmsl$ can be cast as a conic program, c.f.~\cref{Eq:RConicProgram}, that evaluates to a finite value and where strong duality holds.
This completes the proof.
\end{proof}

\subsection{Remark on Assumption~\ref{Assumption:Static}}\label{App:RemarkAssumption2}

A question that follows from the above result is whether one can relax Assumption~\ref{Assumption:Static}  into Assumption~\ref{Asssumption:SufficientCondition}; namely, the existence of a full-rank $\eta_\rmt\in\FRT$.
This is, however, not necessary for Statement~\ref{Condition2}, as we will show using the following counterexample.
Consider a bipartite system ${\rm S = T = AB}$ with equal local dimension $d<\infty$, and ${\bm{\Lambda}} = \{{\rm A,B}\}$ with $\mathcal{F}_{R|{\rm AB}} = \{\ket{\Psi_{\rm AB}^+}\}$, where $\ket{\Psi_{\rm AB}^+}\coloneqq\frac{1}{\sqrt{d}}\sum_{i=0}^{d-1}\ket{i}_{\rm A}\otimes\ket{i}_{\rm B}$ is a maximally entangled state.
This can be understood as the resource theory of athermality at zero temperature where the thermal state for some non-degenerate Hamiltonian is the entangled ground state $\ket{\Psi_{\rm AB}^+}$.
Now, since the single party marginal of $\ket{\Psi_{\rm AB}^+}$ is maximally mixed, one can check that Assumption~\ref{Assumption:Static} is satisfied.
On the other hand, it is clear that the only free state is not full-rank.

\subsection{Proof of Theorem~\ref{Result:DiscriminationTask}}\label{App:Proof-Result:DiscriminationTask}
\begin{proof}
From Theorem~\ref{Result:RTWitness}, $\bmsl\notin\CRTL$ if and only if there exist $\{W_{\rmx}\succeq0\}_{{\rmx}\in\bm{\Lambda}}$ such that
\begin{align}\label{INeq}
\sup_{{\bm{\tau_\Lambda}\in\CRTL}}\sum_{{\rm X}\in\bm{\Lambda}}\tr(\tau_{\rm X}W_{\rm X})<\sum_{{\rm X}\in\bm{\Lambda}}\tr(\sigma_{\rm X}W_{\rm X}).
\end{align}
Without loss of generality, we may assume that for every ${\rmx}\in\bm{\Lambda}$, $W_{\rm X}$ is strictly positive ($W_{\rm X}\succ0$), i.e., having only positive eigenvalues.
This is because we can add $\Delta = \Delta\tr(\tau_{\rm X}) = \Delta\tr(\sigma_{\rm X})$ on both sides and still preserve the strict inequality, where $\Delta>0$ is a positive number.
Now, for each ${\rm X}$, the spectral decomposition of $W_{\rmx}\succ0$ can be written as $W_{\rm X} = \sum_{i=1}^{d_{\rm X}}\omega_{i|{\rm X}}\proj{\psi_{i|{\rm X}}}$, where $d_{\rm X}$ is the dimension of the system ${\rm X}$ and  $\omega_{i|{\rm X}}>0\;\forall\;i,{\rm X}$.
For any set of unitary channels in ${\rm X}$ given by $\mathbfcal{U} = \{\mathcal{U}_{i|{\rm X}}\}_{i=1;{\rm X}\in\bm{\Lambda}}^{d_{\rm X}}$, with $\U_{i|{\rm X}}(\cdot)\coloneqq U_{i|\rmx}(\cdot) U_{i|\rmx}^\dag$ for a given unitary operator $U_{i|\rmx}$, we can write
\begin{align}
\tr\left(\sigma_{\rm X}W_{\rm X}\right) = \sum_{i=1}^{d_{\rm X}}\omega_{i|{\rm X}}\tr\left[\mathcal{U}_{i|{\rm X}}(\sigma_{\rm X})\mathcal{U}_{i|{\rm X}}(\proj{\psi_{i|{\rm X}}})\right]=\frac{1}{d_{\rm X}}\sum_{i=1}^{d_{\rm X}}\tr\left[M_{i|{\rm X}}\mathcal{U}_{i|{\rm X}}(\sigma_{\rm X})\right],
\end{align}
where 
$
M_{i|{\rm X}}\coloneqq d_{\rm X}\omega_{i|{\rm X}}\mathcal{U}_{i|{\rm X}}(\proj{\psi_{i|{\rm X}}}),
$
which is again a non-zero positive semi-definite operator.
From here we obtain
\begin{align}\label{Eq:PDminusIneq}
\sup_{{\bm{\tau_\Lambda}\in\CRTL}}\sum_{{\rm X}\in\bm{\Lambda}}\frac{1}{d_{\rm X}}\sum_{i=1}^{d_{\rm X}}\tr\left[M_{i|{\rm X}}\mathcal{U}_{i|{\rm X}}(\tau_{\rm X})\right]<\sum_{{\rm X}\in\bm{\Lambda}}\frac{1}{d_{\rm X}}\sum_{i=1}^{d_{\rm X}}\tr\left[M_{i|{\rm X}}\mathcal{U}_{i|{\rm X}}(\sigma_{\rm X})\right].
\end{align}
Note that the inequality remains valid if we perform the mapping $M_{i|\rmx} \to \alpha(M_{i|\rmx} + \Delta_0\id)$ for $\alpha, \Delta_0>0$. In particular, for judiciously chosen positive $\alpha<1$, we then have
\begin{align}\label{Eq:POVM1}
M_{i|{\rm X}}\succ0\quad\forall\;i,{\rm X}\quad\&\quad\sum_iM_{i|{\rm X}}\prec\id_{\rm X}\quad\forall\;{\rm X},
\end{align}
thus allowing one to interpret $\{M_{i|{\rm X}}\}_i$, for each ${\rm X}\in\bm{\Lambda}$, as an incomplete POVM.
For any set of states ${\bm{\kappa_\Lambda}} = \{\kappa_{\rm X}\}_{{\rm X}\in\bm{\Lambda}}$ we define the probability of success in the task $D_-$ as:
\begin{align}\label{Eq:PD-}
P_{D_-} ({\bm{\kappa_\Lambda}},\mathbfcal{U})\coloneqq\frac{1}{|{\bm{\Lambda}}|}\sum_{{\rm X}\in\bm{\Lambda}}\frac{1}{d_{\rm X}}\sum_{i=1}^{d_{\rm X}}\tr\left[M_{i|{\rm X}}\mathcal{U}_{i|{\rm X}}(\kappa_{\rmx})\right],
\end{align}
which can be understood as the success probability of using ${\bm{\kappa_\Lambda}}$ to discriminate $\mathbfcal{U}$ in the ``non-deterministic'' task $D_-\coloneqq\left(\{p_{\rm X} = \frac{1}{|{\bm{\Lambda}}|}\},\{p_{i|{\rm X}} = \frac{1}{d_{\rm X}}\},\{M_{i|{\rm X}}\}\right)$.
With this notation, Eq.~\eqref{Eq:PDminusIneq} reads
$
\sup_{\bm{\tau_\Lambda}\in\CRTL}P_{D_-}(\bm{\tau_\Lambda},\mathbfcal{U})<P_{D_-}(\bm{\sigma_\Lambda},\mathbfcal{U}).
$
This means that there exists a finite value $\Delta_1>0$ such that
\begin{align}\label{Eq:Delta}
	P_{D_-}(\bmsl,\mathbfcal{U}) = \Delta_1 + \sup_{\bm{\tau_\Lambda}\in\CRTL}P_{D_-}(\bm{\tau_\Lambda},\mathbfcal{U}).
\end{align}
Now consider the task $D = \left(\{p_{\rm X}\},\{p_{i|{\rm X}}\},\{E_{i|{\rm X}}\}\right)$ with $\mathbfcal{E}=\{\E_{i|{\rm X}}\}_{i=1;{\rm X}\in\bm{\Lambda}}^{d_{\rm X}+1}$ defined as:
\begin{subequations}\label{Eq:D:Components}
\begin{align}
&p_{\rm X} = \frac{1}{|{\bm{\Lambda}}|};\\
&p_{i|{\rm X}} = \frac{1-\epsilon}{d_{\rm X}}\quad{\rm if}\;i\le d_{\rm X}\quad\&\quad p_{d_{\rm X} +1|{\rm X}} = \epsilon;\label{Eq:eps}\\
&\E_{i|{\rm X}} = \mathcal{U}_{i|{\rm X}}\quad{\rm if}\;i\le d_{\rm X}\quad\&\quad \E_{d_{\rm X} +1|{\rm X}} = \calL_{\rm X}\;\\
&E_{i|{\rm X}} = M_{i|{\rm X}}\quad{\rm if}\;i\le d_{\rm X}\quad\&\quad E_{d_{\rm X} +1|{\rm X}} = \id_{\rm X} - \sum_{\rm i=1}^{d_{\rm X}}M_{i|{\rm X}},
\end{align}
\end{subequations}
where $\epsilon\in[0,1]$ is a parameter that will be specified later, and $\calL_{\rm X}$ is an arbitrary channel (hence, we can choose it to be unitary).
From here we learn that $\{E_{i|{\rm X}}\}_{i=1}^{d_{\rm X} +1}$ is a POVM for every ${\rm X}$, which implies that $D$ is a channel-discrimination task.
Furthermore, $D$ is strictly positive when $0<\epsilon<1$ [see also Eq.~\eqref{Eq:POVM1}]. 
Hence, $D$ and $\mathbfcal{E}$ satisfy the description of Theorem~\ref{Result:DiscriminationTask}.

As with \cref{Eq:PD-}, a probability of success can be defined for the task $D$:
\begin{align}\label{Eq:PD}
P_D({\bm{\kappa_\Lambda}},\mathbfcal{E})
\coloneqq\frac{1}{|\bm{\Lambda}|}\sum_{{\rm X}\in\bm{\Lambda}}\sum_{i=1}^{d_{\rm X}+1}p_{i|{\rm X}}\tr\left[E_{i|{\rm X}}\E_{i|{\rm X}}(\kappa_{\rm X})\right]
\end{align}
which can be decomposed as (see also Ref.~\cite{Hsieh2021}):
\begin{align}\label{Eq:Composition}
P_D({\bm{\kappa_\Lambda}},\mathbfcal{E})  = P_{D_-} ({\bm{\kappa_\Lambda}},\mathbfcal{U}) + \epsilon\,\Gamma({\bm{\kappa_\Lambda}},\mathbfcal{E}),
\end{align}
where the second term is defined via Eq.~\eqref{Eq:PD-} and Eq.~\eqref{Eq:D:Components} as:
\begin{align}
\Gamma({\bm{\kappa_\Lambda}},\mathbfcal{E})\coloneqq\frac{1}{|{\bm{\Lambda}}|}\sum_{{{\rm X}\in\bm{\Lambda}}}\tr\left[E_{d_{\rmx}+1|\rmx}\calL_{\rm X}(\kappa_{\rm X})\right]-P_{D_-}({\bm{\kappa_\Lambda}},\mathbfcal{U}).
\end{align}
It then follows from \cref{Eq:Delta} and \cref{Eq:Composition} that
\begin{align}
\sup_{{\bm{\tau_\Lambda}\in\CRTL}}P_{D}({\bm{\tau_\Lambda}},\mathbfcal{E})&\le\sup_{{\bm{\tau_\Lambda}\in\CRTL}}P_{D_-} ({\bm{\tau_\Lambda}},\mathbfcal{U}) + \epsilon\times\sup_{{\bm{\tau_\Lambda}\in\CRTL}}\Gamma({\bm{\tau_\Lambda}},\mathbfcal{E})\nonumber\\
&=P_{D_-}({\bm{\sigma_\Lambda}},\mathbfcal{U}) - \Delta_1 + \epsilon\times\sup_{{\bm{\tau_\Lambda}\in\CRTL}}\Gamma({\bm{\tau_\Lambda}},\mathbfcal{E})\nonumber\\
&=P_D({\bm{\sigma_\Lambda}},\mathbfcal{E}) - \Delta_1 + \epsilon\Delta_2,
\end{align}
where $\Delta_2 \coloneqq \sup_{{\bm{\tau_\Lambda}\in\CRTL}}\Gamma({\bm{\tau_\Lambda}},\mathbfcal{E}) - \Gamma({\bm{\sigma_\Lambda}},\mathbfcal{E})$ is finite since $\Gamma$ is bounded for every $({\bm{\sigma_\Lambda}},\mathbfcal{E})$.
Therefore, we can write
\begin{align}
\sup_{{\bm{\tau_\Lambda}\in\CRTL}}P_{D}({\bm{\tau_\Lambda}},\mathbfcal{E})\le P_D({\bm{\sigma_\Lambda}},\mathbfcal{E}) - \Delta_1 + \epsilon\Delta_2.
\end{align}
Then if $\Delta_2\le0$, we have $\sup_{{\bm{\tau_\Lambda}\in\CRTL}}P_{D}({\bm{\tau_\Lambda}},\mathbfcal{E})< P_D({\bm{\sigma_\Lambda}},\mathbfcal{E})$ $\forall\;\epsilon\in[0,1]$.
When $\Delta_2>0$, one can take $\epsilon<\min\left\{\frac{\Delta_1}{\Delta_2},1\right\}$, c.f.~\cref{Eq:eps}, to conclude that $\sup_{{\bm{\tau_\Lambda}\in\CRTL}}P_{D}({\bm{\tau_\Lambda}},\mathbfcal{E})< P_D({\bm{\sigma_\Lambda}},\mathbfcal{E})$.
The result follows.
\end{proof}

\subsection{Proof of Eq.~(\ref{Eq:GoldenRuleTest})}\label{Proof-Eq:GoldenRuleTest}
\begin{proof}
For every ${\bm{\tau_\Lambda}\in\CRTL}$ and $\mathbfcal{E}_{\bm{\Lambda}}\in\mathfrak{O}_{R|{\rm T}}$, we have $\mathbfcal{E}_{\bm{\Lambda}}(\bm{\tau_\Lambda})\coloneqq\{\mE_{\rm X}(\tau_{\rm X})\}_{{\rm X}\in\bm{\Lambda}}$, where, according to the definitions,
\begin{align}
&\exists\;\E_\rms\;{\rm compatible\;with}\;\mathbfcal{E}_{\bm{\Lambda}}\;{\rm s.t.}\;\ttr_{\rm S\setminus T\to S\setminus T}\E_\rms\in\mathcal{O}_{R|{\rm T}};\\
&\exists\,\eta_\rms\;{\rm compatible\;with}\;{\bm{\tau_\Lambda}}\;{\rm s.t.}\;\tr_{\rm S\setminus T}(\eta_\rms)\in\FRT.
\end{align}
In other words, we have $\ttr_{\rm S\setminus X\to S\setminus X}\E_\rms = \mE_{\rm X}$ and $\tr_{\rm S\setminus X}(\eta_\rms) = \tau_{\rm X}$ for every ${\rm X}\in\bm{\Lambda}$.
This means that
\begin{align}
&\tr_{\rm S\setminus X}\left[\E_\rms(\eta_\rms)\right] =\mE_{\rm X}\left(\tau_{\rm X}\right)\quad\forall\; {\rm X}\in\bm{\Lambda};\\
&\tr_{\rm S\setminus T}\left[\E_\rms(\eta_\rms)\right] = \ttr_{\rm S\setminus T\to S\setminus T}\E_\rms\left[\tr_{\rm S\setminus T}(\eta_\rms)\right]\in\FRT,
\end{align}
where we have used Eq.~\eqref{Eq:Tr}.
Thus, there exists a global state, $\E_\rms(\eta_\rms)$, whose marginal state in $\rmt$, $\tr_{\rm S\setminus T}\left[\E_\rms(\eta_\rms)\right]$, is free, such that it is compatible with $\mathbfcal{E}_{\bm{\Lambda}}(\bm{\tau_\Lambda})$.
Hence, we conclude that $\mathbfcal{E}_{\bm{\Lambda}}(\bm{\tau_\Lambda})\in\CRTL$.
\end{proof}

\subsection{Proof of Theorem~\ref{Result:R-CompatibilityMonotone}}\label{Proof-Result:R-CompatibilityMonotone}

\begin{proof}
The first statement follows directly from the definition of $\RRT$, it thus suffices to prove the second statement.
For every $\mathbfcal{E}_{\bm{\Lambda}}\in\mathfrak{O}_{R|{\rm T}}$ and $\bmsl$, we have
\begin{align}
\RRT(\bmsl)& = \inf_{\lambda, \bm{\kappa_\Lambda}} \log_2\left\{\lambda\,|\,\sigma_{\rm X}\preceq\lambda\kappa_{\rm X}\;\forall\;{\rm X}\in\bm{\Lambda};\bm{\kappa_\Lambda}\in\CRTL\right\}\nonumber\\
&\ge\inf_{\lambda, \bm{\kappa_\Lambda}}\log_2\left\{\lambda\,|\,0\preceq\mE_{\rm X}\left(\lambda\kappa_{\rm X} - \sigma_{\rm X}\right)\;\forall\;{\rm X}\in\bm{\Lambda};\bm{\kappa_\Lambda}\in\CRTL\right\}\nonumber\\
&\ge\inf_{\lambda, \bm{\eta_\Lambda}}\log_2\left\{\lambda\,|\,0\preceq\lambda\eta_{\rm X} - \mE_{\rm X}\left(\sigma_{\rm X}\right)\;\forall\;{\rm X}\in\bm{\Lambda};\bm{\eta_\Lambda}\in\CRTL\right\}\nonumber\\
&=\RRT\left[\mathbfcal{E}_{\bm{\Lambda}}(\bm{\sigma_\Lambda})\right].
\end{align}
The first inequality holds because each $\mE_{\rmx}$ maps positive operators (and possibly some non-positive operators) to positive operators, thereby making the range of minimization larger. The second inequality follows from the linearity of $\mE_{\rmx}$ and the fact that $\mathbfcal{E}_{\bm{\Lambda}}\in\mathfrak{O}_{R|{\rm T}}$ implies $\{\mE_{\rm X}(\kappa_{\rm X})\}_{{\rm X}\in\bm{\Lambda}}\in\CRTL$ for every $\{\kappa_{\rm X}\}_{{\rm X}\in\bm{\Lambda}}\in\CRTL$ [i.e., Eq.~\eqref{Eq:GoldenRuleTest}], thereby (potentially) making the minimization range larger.
\end{proof}

\subsection{Proof of Theorem~\ref{Result:Convertibility}}\label{App:Proof-Result:Convertibility}
\begin{proof}
We first show that Statement~\ref{Free_op_transformation} implies Statement~\ref{Psucc_allD}.
For every given $D$, a direct computation shows that
\begin{align}
	\Scr\left(\bm{\sigma}_{\bm{\Lambda}},{D}\right)& = \Scr\left[\mathbfcal{E}_{\bm{\Lambda}}\left(\bm{\tau}_{\bm{\Lambda}}\right),{D}\right]\coloneqq\sup_{\mathbfcal{E}_{\bm{\Lambda}}'\in\mathfrak{O}_{R|\rmt,\bm{\Lambda}}}\sum_{{\rm X},i}p_{\rm X}p_{i|{\rm X}}{\rm tr}\left[E_{i|{\rm X}}\left(\mathcal{E}_{\rm X}'\circ\mathcal{E}_{\rm X}\right)(\tau_{\rm X})\right]\nonumber\\
	&\le\sup_{\mathbfcal{E}_{\bm{\Lambda}'}\in\mathfrak{O}_{R|\rmt,\bm{\Lambda}}}\sum_{{\rm X},i}p_{\rm X}p_{i|{\rm X}}{\rm tr}\left[E_{i|{\rm X}}\mathcal{E}_{\rm X}'(\tau_{\rm X})\right] = \Scr(\bm{\tau}_{\bm{\Lambda}},{D}),
\end{align}
where we have used $\mathbfcal{E}_{\bm{\Lambda}}'\circ\mathbfcal{E}_{\bm{\Lambda}}\in\mathfrak{O}_{R|\rmt,\bm{\Lambda}}$ if $\mathbfcal{E}_{\bm{\Lambda}}',\mathbfcal{E}_{\bm{\Lambda}}\in\mathfrak{O}_{R|\rmt,\bm{\Lambda}}$. 

It remains to show that Statement~\ref{Psucc_allD} implies Statement~\ref{Free_op_transformation}.
First, we note that Statement~\ref{Psucc_allD} can be rewritten as
\begin{align}\label{Eq:non-negative-computation01}
	0&\le \inf_{{D}}\left[\Scr(\bm{\tau}_{\bm{\Lambda}},{D})- \Scr(\bm{\sigma}_{\bm{\Lambda}},{D})\right]\le \inf_{{D}}\left[\Scr(\bm{\tau}_{\bm{\Lambda}},{D})- \sum_{{\rm X},i}p_{\rm X}p_{i|{\rm X}}{\rm tr}\left(E_{i|{\rm X}}\sigma_{\rm X}\right)\right]\nonumber\\
	&=\inf_{{D}}f({D}),
\end{align}
where the minimization is taken over all non-deterministic $D$, 
\begin{align}
f({D})\coloneqq\sup_{\mathbfcal{E}_{\bm{\Lambda}}\in\mathfrak{O}_{R|\rmt,\bm{\Lambda}}}\sum_{{\rm X},i}p_{\rm X}p_{i|{\rm X}}{\rm tr}\left[E_{i|{\rm X}}\left(\mathcal{E}_{\rm X}(\tau_{\rm X})-\sigma_{\rm X}\right)\right],
\end{align}
and we have used the fact that identity map is a free operation in arriving at the second inequality.
Consider another function
\begin{align}
g\left(\{L_{\rm X}\}_{\rm X}\right)\coloneqq\max_{\mathbfcal{E}_{\bm{\Lambda}}\in\mathfrak{O}_{R|\rmt,\bm{\Lambda}}}\sum_{{\rm X}}{\rm tr}\left[L_{\rm X}\left(\mathcal{E}_{\rm X}(\tau_{\rm X})-\sigma_{\rm X}\right)\right].
\end{align}
Then one can check that Eq.~\eqref{Eq:non-negative-computation01} holds if and only if 
\begin{align}\label{Eq:non-negative-computation02}
0\le\min_{\substack{0\le L_{\rm X}\\\sum_{\rm X}{\rm tr}(L_{\rm X})\le d_{\rm max}}}g\left(\{L_{\rm X}\}_{\rm X}\right),
\end{align} 
where $d_{\rm max}\coloneqq\max_{\rm X}d_{\rm X}$.
To see why this is true, consider $L_{\rm X} = \sum_ip_{\rm X}p_{i|{\rm X}}E_{i|{\rm X}}$, we have $\min_{\substack{0\le L_{\rm X}\\\sum_{\rm X}{\rm tr}(L_{\rm X})\le d_{\rm max}}}g\left(\{L_{\rm X}\}_{\rm X}\right)\le\inf_{{D}}f({D})$.
Conversely, for every feasible $\{L_{\rm X}\}_{\rm X}$ of the minimization in Eq.~\eqref{Eq:non-negative-computation02}, one can define
\begin{align}
E_{0|{\rm X}} = \frac{L_{\rm X}}{\max_{\rm X}\norm{L_{\rm X}}_\infty+\delta}\quad\&\quad p_{\rm X} = \frac{1}{|{\bm{\Lambda}}|},
\end{align}
where $\delta>0$ is some positive value.
Then one can see that 
\begin{align}
\max_{\mathbfcal{E}_{\bm{\Lambda}}\in\mathfrak{O}_{R|\rmt,\bm{\Lambda}}}\sum_{{\rm X}}\frac{{\rm tr}\left[L_{\rm X}\left(\mathcal{E}_{\rm X}(\tau_{\rm X})-\sigma_{\rm X}\right)\right]}{|{\bm{\Lambda}}|(\max_{\rm X}\norm{L_{\rm X}}_\infty+\delta)}\ge\inf_{{D}}f({D})
\end{align} 
for every feasible $\{L_{\rm X}\}_{\rm X}$, which implies Eq.~\eqref{Eq:non-negative-computation02}.

Using the above observation, we apply Sion's minimax theorem to Eq.~\eqref{Eq:non-negative-computation02} (see also, e.g., Ref.~\cite{Paul}) and obtain
\begin{align}
0\le\max_{\mathbfcal{E}_{\bm{\Lambda}}\in\mathfrak{O}_{R|\rmt,\bm{\Lambda}}}\min_{\substack{0\le L_{\rm X}\\\sum_{\rm X}{\rm tr}(L_{\rm X})\le d_{\rm max}}}\sum_{{\rm X}}{\rm tr}\left[L_{\rm X}\left(\mathcal{E}_{\rm X}(\tau_{\rm X})-\sigma_{\rm X}\right)\right].
\end{align}
Let $\mathbfcal{E}_{\bm{\Lambda}}^{\rm opt}$ be the one achieving the maximisation.
Then we have
\begin{align}
0\le\min_{\substack{0\le L_{\rm X}\\\sum_{\rm X}{\rm tr}(L_{\rm X})\le d_{\rm max}}}\sum_{{\rm X}}{\rm tr}\left[L_{\rm X}\left(\mathcal{E}_{\rm X}^{\rm opt}(\tau_{\rm X})-\sigma_{\rm X}\right)\right].
\end{align}
This means
$
\mathcal{E}_{\rm X}^{\rm opt}(\tau_{\rm X})\ge\sigma_{\rm X}
$
for every ${\rm X}$.
Since a semi-definite positive trace-less operator must be a zero operator, we conclude that
$
\mathbfcal{E}_{\bm{\Lambda}}^{\rm opt}(\bm{\tau}_{\bm{\Lambda}}) = \bm{\sigma}_{\bm{\Lambda}}.
$
The proof is thus completed.
\end{proof}

\subsection{POVM Elements for the Numerical Example}\label{App:ExamplePOVMs}
For each $\rmx\in\{{\rm AB, AC}\}$, the POVM elements that we use in the example in the main text are constructed from the operators 
$M_{i|{\rm X}}\coloneqq d_{\rm X}\omega_{i|{\rm X}}\mathcal{U}_{i|{\rm X}}(\proj{\psi_{i|{\rm X}}})$ 
where (with respect to the computational basis $\{\ket{00},\ket{01},\ket{10},\ket{11}\}$)
\begin{align}
    \omega_{1|{\rm X}} &= 0.010000027026545,\quad  
    \omega_{2|{\rm X}} = 0.010000058075968, \nonumber  \\
    \omega_{3|{\rm X}} &= 0.458638621962197,\quad
    \omega_{4|{\rm X}} = 0.537143367559183.
\end{align}
and
\begin{align}
    \ket{\psi_{1|{\rm X}}} &= 0.668877697040469 \ket{01} -0.743372468148935 \ket{10}, \nonumber \\
    \ket{\psi_{2|{\rm X}}} &= \ket{11}, \quad
    \ket{\psi_{3|{\rm X}}} = \ket{00}, \nonumber \\
    \ket{\psi_{4|{\rm X}}} &= 0.743372468148935 \ket{01} + 0.668877697040469 \ket{10}.
\end{align}
These are taken from the eigenstates and eigenvalues of $W_{\rm X} + 0.01\times\id_{\rm X}$, where $\{W_{\rm X}\}_{\rm X}$ are the witness operators from Theorem~\ref{Result:RTWitness} that can be found by SDP (using the terminology in the proof of Theorem~\ref{Result:DiscriminationTask}, it also means that we choose $\Delta = 0.01$).
The actual POVMs used in the calculation are given by
\begin{subequations}
\begin{align}
&E_{i|{\rm X}} = \frac{M_{i|{\rm X}} + 0.01\times\id_{\rm X}}{\mu_{\rm X} + 0.01}\;\;{\rm for}\;i=1,2,3,4\quad\&\quad E_{5|{\rm X}} = \id_{\rm X} - \sum_{i=1}^{4} E_{i|{\rm X}},
\end{align}
\end{subequations} 
where $\mu_{\rm X} = {\lVert \sum_{i=1}^{4} M_{i|{\rm X}} + 0.01\times\id_{\rm X}  \rVert_{\infty}}$.
These POVM elements are obtained following the construction given in the proof of \cref{Result:DiscriminationTask}.
The witness operators $W_{\bm{\Lambda}}$, in turn, were obtained by solving a semidefinite program (see Ref.~\cite{Tabia}) that can be used to certify the entanglement transitivity of $\bm{\sigma}_{\bm{\Lambda}}^{W}$.

\end{document}